\documentclass{LMCS}
\usepackage{etex}
\usepackage{xcolor}
\usepackage{natbib}
\usepackage{url}
\usepackage{graphicx}
\usepackage{amsmath,amssymb,amsfonts,amsthm}
\usepackage{mathrsfs,stmaryrd}
\usepackage{bussproofs}
\usepackage{hyperref,enumerate}
\usepackage[all]{xy}
\UseComputerModernTips

\theoremstyle{plain}
\newtheorem{theorem}{Theorem}[section]

\newtheorem{lemma}[theorem]{Lemma}

\theoremstyle{definition}
\newtheorem{definition}[theorem]{Definition}
\newtheorem{example}[theorem]{Example}

\theoremstyle{remark}
\newtheorem{remark}[theorem]{Remark}
\newtheorem*{notation}
{Notation}

\newcommand{\hide}[1]{}

\newcommand{\ie}{\textit{i.e.}}
\newcommand{\cf}{\textit{cf.}}

\newcommand{\eg}{\textit{e.g.}}

\newcommand{\opr}{\mathsf{o}}

\newcommand{\U}[1]{|{#1}|}

\newcommand{\llrrbrk}[1]{{\llbracket{#1}\rrbracket}}
\newcommand{\lrangle}[1]{{\langle{#1}\rangle}}

\newcommand{\ol}[1]{\overline{#1}}
\newcommand{\ul}[1]{\underline{#1}}
\newcommand{\wt}[1]{\widetilde{#1}}
\newcommand{\wh}[1]{\widehat{#1}}

\newcommand{\setof}[1]{\{\, #1 \,\}}
\newcommand{\setofz}[1]{\{ #1 \}}
\newcommand{\suchthat}{\mid}

\newcommand{\inj}[1]{\iota_{#1}}

\newcommand{\arrowlength}{4.25}
\renewcommand{\mapsto}{\mathrel{\xy\ar@{|->}(\arrowlength,0)\endxy}}
\newcommand{\natrightarrow}{\mathrel{\xy\ar(\arrowlength,0)^(.4){.}\endxy}}
\renewcommand{\rightarrow}{\mathrel{\xy\ar@{->}(\arrowlength,0)\endxy}}
\renewcommand{\leftarrow}{\mathrel{\xy\ar@{->}(-\arrowlength,0)\endxy}}
\newcommand{\epirightarrow}{\mathrel{\xy\ar@{->>}(\arrowlength,0)\endxy}}
\newcommand{\epileftarrow}{\mathrel{\xy\ar@{<<-}(\arrowlength,0)\endxy}}
\newcommand{\rightembedding}{\mathrel{\xy\ar@{^(->}(\arrowlength,0)\endxy}}

\newcommand{\catalg}[1]{{#1\text{-}\mathbf{Alg}}}
\newcommand{\catAlg}[2]{{{#1}^{#2}}}

\newcommand{\cat}[1]{\mathscr{#1}}
\newcommand{\iso}{\cong}
\newcommand{\tensor}{\otimes}
\newcommand{\cotensor}%
{\mathrel{\xy\ar@{}(0,0)|-{\textstyle\cap}\ar@{}(0,0)|-{\textstyle|}\endxy}}
\newcommand{\id}{\mathrm{id}}

\newcommand{\comp}{\circ} 

\newcommand{\monad}[1]{\mathbf{#1}}

\newcommand{\monten}{\cdot}

\newcommand{\alphaact}{\wt{\alpha}}
\newcommand{\lambdaact}{\wt{\lambda}}

\newcommand{\s}[1]{\mathsf{#1}}
\newcommand{\mbf}[1]{\mathbf{#1}}
\newcommand{\mbs}[1]{{\boldsymbol{#1}}}
\newcommand{\st}{\mathsf{st}}
\newcommand{\qt}{\mathsf{q}}

\newcommand{\Set}{\mathbf{Set}}
\newcommand{\Nat}{\mathbb{N}}

\newcommand{\Nom}{\mathbf{Nom}}

\newcommand{\SetToI}{\Set^\ISet}
\newcommand{\ISet}{\mathbb{I}}

\newcommand{\algterm}[3]%
{#1 \vartriangleright #2 \vdash #3}
\newcommand{\funterm}[4]%
{#1 : #2 \vartriangleright #3 \vdash #4}
\newcommand{\funtermz}[2]%
{#1 \vdash #2}
\newcommand{\eqsys}[5]%
{\funterm{#1}{#2}{#3}{#4 \equiv #5}}
\newcommand{\funeq}[4]%
{\algterm{#1}{#2}{#3 \equiv #4}}
\newcommand{\eqsysz}[3]%
{\funtermz{#1}{#2 \equiv #3}}
\newcommand{\terminctx}[2]%
{#1 \vdash #2}
\newcommand{\eqn}[3]%
{#1 \vdash #2 \equiv #3}
\newcommand{\Eeqn}[3]%
{#1 \vdash_E #2 \equiv #3}
\newcommand{\eqnz}[2]%
{#1 \equiv #2}

\newcommand{\algstr}[2]{#1_{#2}}

\newcommand{\toTES}[1]{\ul{#1}}

\newcommand{\tensorext}[2]{\langle{#1}\rangle{#2}}

\newcommand{\Perm}{\mathfrak{S}_0(\Atom)}

\newcommand{\Atom}{\mathsf{A}}
\newcommand{\NomAtom}{\mathbb{A}}

\newcommand{\permact}{\cdot}
\DeclareMathOperator{\freshfor}{\#}

\newcommand{\suppsym}[1]{\mathsf{supp}_{#1}}
\newcommand{\supp}[2]{\suppsym{#1}(#2)}
\newcommand{\suppz}[1]{\supp{}{#1}}
\newcommand{\septen}{\freshfor}

\newcommand{\sephom}[1]{[ #1 ]}
\newcommand{\sephombig}[1]{\big[ #1 \big]}
\newcommand{\atmtrans}[2]{({#1 \; #2})}
\newcommand{\atmabs}[1]{\lrangle{#1}\,}
\newcommand{\atmprd}[1]{\NomAtom^{\septen #1}}
\newcommand{\atmhom}[2]{\sephom{\atmprd{#1},#2}}
\newcommand{\atmlstz}[1]{\boldsymbol{#1}}
\newcommand{\atmlst}[2]{{\atmlstz{#1}}^{#2}}
\newcommand{\atmlsth}[2]{\boldsymbol{#1}}
\newcommand{\nomctx}[2]{[{#1}]{#2}}
\newcommand{\Vext}[2]{{#1}^{\langle#2\rangle}}
\newcommand{\ev}{\mathsf{e}}
\newcommand{\subst}[1]{\{ #1 \}}

\newcommand{\myproof}[1]{\mbox{$ #1 \DisplayProof $}}
\newcommand{\mh}[2]{\mbox{$#1$-#2}}
\newcommand{\nh}[2]{\mbox{{#1}-#2}}

\newcommand{\Ax}{\mathcal A}

\newcommand{\inthom}[1]{[ #1 ]}

\usepackage{enumerate}
\usepackage{hyperref}

\theoremstyle{plain}
\def\eg{{\em e.g.}}
\def\cf{{\em cf.}}

\def\doi{6 (3:12) 2010}
\lmcsheading%
{\doi}
{1--24}
{}
{}
{May\phantom.15, 2010}
{Sep.~\phantom08, 2011}
{}

\begin{document}

\title[On the Mathematical Synthesis of Equational Logics]
{On the Mathematical Synthesis of Equational Logics\rsuper*}

\author[M.\ Fiore]{Marcelo Fiore\rsuper a}	
\address{{\lsuper a}University of Cambridge, Computer Laboratory,\newline 15 JJ Thomson Avenue,
  Cambridge CB3 0FD, UK.}	
\email{Marcelo.Fiore@cl.cam.ac.uk}  

\author[C.-K.\ Hur]{Chung-Kil Hur\rsuper b}	
\address{
{\lsuper b}Max Planck Institute for Software Systems~(MPI-SWS),\newline Kaiserslautern
and Saarbr\"ucken, Germany.} 
\email{
gil@mpi-sws.org}  


%
%
\titlecomment{{\lsuper*}This paper gives a complete development (with proofs) of
results announced in~\citet{FioreHur08}.  However, for simplicity of
exposition, these results are restricted here to the case of
\mbox{\emph{mono-sorted}} algebraic theories and equational logics.}
%

\maketitle

\begin{abstract}
We provide a mathematical theory and methodology for synthesising
equational logics from algebraic metatheories.  We illustrate our
methodology by means of two applications: a rational reconstruction of
Birkhoff's Equational Logic and a new equational logic for reasoning about
algebraic structure with name-binding operators.
\end{abstract}

\section*{\ul{\large Introduction}}\label{S:introduction}

\citet{Birkhoff35} initiated the general study of algebraic structure.
Importantly for our concerns here, his development was from (universal)
algebra to (equational) logic.  Birkhoff's starting point was the informal
conception of algebra based on familiar concrete examples.  Abstracting
from these, he introduced the concepts of signature and equational
presentation, and thereby formalised what is now our notion of (abstract)
algebra.  Subsequently 
he set up the model theory of equational presentations~(varieties) and
analysed their structure from the standpoint of 
logical inference 
for algebraic languages.  In doing so, he introduced Equational Logic as a
sound and complete deductive system for reasoning about equational
assertions in Algebraic Theories.

Since Birkhoff's work, our understanding of algebraic structure has
deepened; having been both systematised and extended
(see~\eg~\citet{Lawvere63}, \citet{Ehresmann}, \citet{Burroni},
\citet{KellyPower93}, \citet{Power99}).  On the other hand, the
development of equational logics
has remained ad~hoc.  The main aim 
of the current work is to 
fill in this gap.  


Our standpoint is that equational logics should arise from algebraic
structure.  In this direction, our first purpose 
is to provide a mathematical theory and methodology for synthesising
equational logics from algebraic metatheories~(
Part~I).  Our second
purpose 
is to establish 
the practicality of the approach.  
In this respect, 
we illustrate our methodology by means of two applications: a rational
reconstruction of Birkhoff's Equational Logic and a new equational logic
for reasoning about algebraic structure with name-binding operators~(
Part~II).

\section*{\ul{\large Part I. Theory}}

In this first part of the paper, we 
present our mathematical framework for synthesising equational logics.
For simplicity of exposition, we restrict attention to the
\mbox{\emph{mono-sorted}} context.  As such, we consider algebraic
metatheories given by strong monads on symmetric monoidal categories.
These provide algebraic structure that allows the specification of
equational presentations in the form of \emph{Monadic Equational
Systems}~(Section~\ref{MonadicEquationalSystems}).
Monadic Equational Systems come equipped with a canonical model theory
whereby models are Eilenberg-Moore algebras satisfying the
equations. 
An \emph{Equational Metalogic}~(Section~\ref{EquationalMetalogic}) for
reasoning about equality in such models is presented.  This deductive
system has been designed to guarantee sound 
derivations. 
As for completeness~(Section~\ref{InternalCompletenessSection}), a
mathematical justification of the well-known use of free constructions in
equational completeness proofs is 
given, 
and this 
is backed up with an inductive method for constructing free
algebras. 

\section{Monadic Equational Systems}
\label{MonadicEquationalSystems}

Monadic Equational Systems~(MESs) are defined and their model theory is
explained.  

\subsection{Monadic Equational Systems}

The concept of MES provides a general abstract notion of equational
presentation.

\begin{definition}[Terms and equations]
  A \emph{term} for an endofunctor $T$ on a category $\cat C$ of arity $A$ and
  coarity $C$ is a Kleisli map $C\rightarrow T A$ in $\cat C$.  A parallel
  pair of terms ${\eqnz{t}{t'}:C\rightarrow T A}$ is called an
  \emph{equation}.
\end{definition}

\begin{definition}[Monadic Equational Systems]
  A \emph{MES} $\mathbb{S} = (\ul{\cat C},\monad{T},\Ax)$ consists of a
  strong monad $\monad T$ on a symmetric monoidal closed category
  $\ul{\cat C}$ together with a set of equations $\Ax$.
\end{definition}

\begin{notation}
For a strong monad~$\monad T$ on a symmetric monoidal closed 
category~$\ul{\cat C}$ we implicitly assume that the respective underlying
structures are denoted by ${(T,\eta,\mu,\st)}$ and 
${(\cat C,I,\tensor,[-,=])}$.
\end{notation}


\subsection{Model theory}

Terms admit interpretations in algebras and these give a model-theoretic
notion of equality.

\medskip

Let $(T,\st)$ be a strong endofunctor on a symmetric monoidal closed 
category~$\ul{\cat C}$.  Every term~${t:C\rightarrow TA}$ admits an
\emph{interpretation}
\[
\llrrbrk{t}_{(X,s)}: [A,X]\tensor C \rightarrow X
\]
in a \mbox{$T$-algebra} $(X,s:TX\rightarrow X)$ given by the composite
\[
\xymatrix@C=50pt@R=0pt{
[A,X]\tensor C \ar[r]^-{[A,X]\tensor t} & 
[A,X]\tensor TA
\ar[r]^-{\st_{[A,X],A}}
&
T\big([A,X]\tensor A\big) 
}
\!\!\!
\xymatrix@C=30pt@R=0pt{
\ar[r]^-{T(\epsilon^A_X)} 
& 
TX \ar[r]^-{s} 
& 
X
}.
\]
We thus obtain a \emph{satisfaction relation} between algebras and equations:
for all \mbox{$T$-algebras}~$(X,s)$ and equations 
$\eqnz u v: C\rightarrow TA$,
\[
(X,s)\models \eqnz u v: C\rightarrow TA
\quad \text{iff} \quad 
\llrrbrk u_{(X,s)} = \llrrbrk v_{(X,s)}
 : [A,X]\tensor C\rightarrow X
\enspace .
\]

\begin{definition}[Algebras]
An \emph{\mh{\mathbb{S}} algebra} for a 
MES~$\mathbb{S} =(\ul{\cat C},\monad T,\Ax)$ is an Eilenberg-Moore
algebra~$(X,s)$ for the monad~$\monad{T}$ satisfying the equations in $\Ax$;
that is, such that\linebreak 
${(X,s) \models \eqnz u v: C\rightarrow TA}$ for all 
${(\eqnz u v: C\rightarrow TA)\in \Ax}$.
\end{definition}

The category~$\catalg{\mathbb{S}}$ is the full subcategory of the
category~$\catAlg{\cat C}{\monad T}$ (of Eilenberg-Moore algebras for the
monad~$\monad T)$ consisting of the $\mathbb{S}$-algebras. 
We thus have the following situation
\[
\xymatrix{
\catalg{\mathbb{S}}\
\ar@{^(->}[r] \ar[dr]_-{U_\mathbb{S}}
& \catAlg{\cat C}{\monad T}\ \ar[d]^(.5){U_{\monad T}} 
\ar@{^(->}[r] 
& \catalg{T} \ar[dl]^-{U_T}
\\
& \cat C & 
}
\]
where $\catalg T$ denotes the category of algebras for the
endofunctor~$T$.

\section{Equational Metalogic}
\label{EquationalMetalogic}

We present a sound deductive system for reasoning about the equality of terms
in MESs.

\subsection{Equational Metalogic}

The \emph{Equational Metalogic~(EML)} associated to a 
MES~$\mathbb S=(\ul{\cat C},\monad T,\Ax)$ has judgements of the form 
\[
\eqn{\Ax}{u}{v}:C\rightarrow TA
\enspace,
\]
where $u$ and $v$ are terms of arity $A$ and coarity $C$, and consists of the
following inference rules.
\begin{enumerate}[$\bullet$]
\item
Equality rules.\\
\[
\myproof{
\AxiomC{$\phantom{\Ax}$}
\LeftLabel{$\s{Ref}$}
\UnaryInfC{$\eqn \Ax u u : C\rightarrow TA$}
}
\qquad\qquad\qquad
\myproof{
\AxiomC{$\eqn \Ax u v : C\rightarrow TA$}
\LeftLabel{$\s{Sym}$}
\UnaryInfC{$\eqn \Ax v u : C\rightarrow TA$}
}
\]\\[-2.5mm]

\[
\myproof{
\AxiomC{
  $\eqn \Ax u v: C\rightarrow TA$
  \qquad 
  $\eqn \Ax v w: C\rightarrow TA$}
\LeftLabel{$\s{Trans}$}
\UnaryInfC{
  $\eqn \Ax u w: C\rightarrow TA$}
}
\]\\[-2.5mm]

\item
Axioms.\\
\[
\myproof{
\LeftLabel{$\s{Axiom}$}
\AxiomC{$(\eqnz u v: C\rightarrow TA)\in \Ax$}
\UnaryInfC{$
\eqn \Ax u v : C\rightarrow TA
$}
}
\]\\[-2.5mm]

\item
Congruence of substitution.\\
\[
\myproof{
\AxiomC{$
\eqn{\Ax}{u_1}{v_1}:C \rightarrow TB
\qquad
\eqn{\Ax}{u_2}{v_2}:B \rightarrow TA
$}
\LeftLabel{$\s{Subst}$}
\UnaryInfC{$
\eqn{\Ax}{u_1\subst{u_2}\,}{\, v_1\subst{v_2}}:C\rightarrow TA$}
}
\]
where $w_1\subst{w_2}$ denotes the Kleisli composite
$\!\!\xymatrix@R=0pt@C=7.5pt{ 
  C\ar[rr]^-{w_1} && TB \ar[rrr]^-{T(w_2)} &&& T(TA) \ar[rr]^-{\mu_A} &&
  TA }$.\\

\item
Congruence of tensor extension.\\
\[
\myproof{
\AxiomC{$
\eqn{\Ax}{u}{v}: C\rightarrow TA
$}
\LeftLabel{$\s{Ext}$}
\RightLabel{}
\UnaryInfC{$
\eqn{\Ax}{\tensorext{V}{u} \,}{\, \tensorext{V}{v}}
  : V\tensor C\rightarrow T(V\tensor A)
$}
}
\]
where 
$\tensorext{V}{w}$ denotes the composite
$\!\!\xymatrix@R=0pt@C=30pt{
V\tensor C \ar[r]^-{V\tensor w} 
& 
V\tensor TA \ar[r]^-{\st_{V,A}} 
& 
T(V\tensor A)
}$.\\

\item
Local character.\\
\[\quad\qquad
\myproof{
\AxiomC{$\eqn{\Ax}{ u\comp e_i \,}{\, v\comp e_i}
    : C_i\rightarrow TA \quad (i\in I)$}
\LeftLabel{$\s{Local}$}
\RightLabel{\big($\setof{e_i:C_i\rightarrow C}_{i\in I}$ jointly epi\big)}
\UnaryInfC{$\eqn \Ax u  v : C\rightarrow TA$}
}
\]\\
(Recall that a family of maps $\setof{e_i:C_i\rightarrow C}_{i\in I}$ is said
to be \emph{jointly epi} if, for any $f,g:C\rightarrow X$ such that
$\forall_{i\in I}\ 
 {f\comp e_i \, = \, g\comp e_i: C_i\rightarrow X}$,
it follows that $f\,=\,g$.)\\
\end{enumerate}

\begin{remark}
In the presence of coproducts and under the rule~$\s{Ref}$, the
rules~$\s{Subst}$ and~$\s{Local}$ are inter-derivable with the
rules\\[-1mm]
\[\myproof{
\AxiomC{$
\eqn{\Ax}{u}{v}: C\rightarrow T\big(\coprod_{i\in I}B_i\big)
\qquad
\eqn{\Ax}{u_i}{v_i}: B_i\rightarrow TA \enspace (i\in I)$}
\LeftLabel{$\s{Subst}_\amalg$}
\UnaryInfC{$\eqn\Ax{u\subst{[u_i]_{i\in I}}}{v\subst{[v_i]_{i\in I}}}
:C\rightarrow TA$}
}\]\\
and\\[-2mm]
\[\myproof{
\AxiomC{$\eqn{\Ax}{u\comp e}{v \comp e}: C'\rightarrow TA$}
\LeftLabel{$\s{Local}_1$}
\RightLabel{($e:C'\epirightarrow C$ epi)}
\UnaryInfC{$\eqn \Ax u  v : C\rightarrow TA$}
}\]\\
Indeed, consider the rule\\[-1mm]
\[\myproof{
\AxiomC{$\eqn{\Ax}{u_i}{v_i}: C_i\rightarrow TA \quad (i\in I)$}
\LeftLabel{$\s{Local}_\amalg$}
\UnaryInfC{$\eqn\Ax{[u_i]_{i\in I}}{[v_i]_{i\in I}}
:\coprod_{i\in I}C_i\rightarrow TA$}
}\]\\[.5mm]
and note that: 
$(i)$~the rule~$\s{Local}$ is derivable from the rules~$\s{Local}_1$
and~$\s{Local}_\amalg$, which are in turn instances of the
rule~$\s{Local}$; 
$(ii)$~the rule~$\s{Subst}_\amalg$ is derivable from the rules~$\s{Subst}$
and $\s{Local}_\amalg$;
$(iii)$~the rule~$\s{Subst}$ is an instance of the
rule~$\s{Subst}_\amalg$; and 
$(iv)$~assuming the rule~$\s{Ref}$, the rule~$\s{Local}_\amalg$ is
derivable from the rule~$\s{Subst}_\amalg$.
\end{remark}

\subsection{Soundness}

The following result states the soundness of derivability in EML.
We write $\catalg{\mathbb S}\models \eqnz u v: C\rightarrow TA$ whenever
$(X,s)\models \eqnz u v : C\rightarrow TA$ for all 
$\mathbb S$-algebras~$(X,s)$.

\begin{theorem}[Soundness]\label{EMLsoundness}
For a MES~$\mathbb S=(\ul{\cat C},\monad T,\Ax)$,
\[
\eqn \Ax u v : C\rightarrow TA
\ \ \text{ implies }\ \ 
\catalg{\mathbb{S}}\models \eqnz u v : C\rightarrow TA
\enspace .
\]
\end{theorem}
\proof 
See Appendix~\ref{EMLsoundnessproof}.
\qed 

\section{Internal Completeness}
\label{InternalCompletenessSection}

In this section we build a mathematical basis for investigating completeness.
Our main tools are an internal completeness result for MESs that admit free
algebras together with an inductive method for constructing them. 

\subsection{Internal Completeness}
Let $\mathbb S=(\ul{\cat C},\monad T,\Ax)$ be a MES admitting free algebras;
that is, such that the forgetful 
functor~${U_\mathbb{S}:\catalg{\mathbb S}\rightarrow \cat C}$ has a left
adjoint.  We denote the free \mh{\mathbb S} algebra on ${X\in\cat C}$ as 
$(T_{\mathbb S}X,
  \tau^{\mathbb S}_X: TT_{\mathbb S}X\rightarrow T_\mathbb{S}X)$,
and the associated \emph{free \mh{\mathbb S}algebra monad} as 
${\monad T}_{\mathbb S} = (T_{\mathbb S}, \eta^{\mathbb S}, \mu^{\mathbb S})$.
Then, the 
embedding~$\catalg{\mathbb S}\, \rightembedding \cat C^{\monad T}$ induces
a strong monad 
morphism~${\qt^\mathbb{S}:\monad T\rightarrow \monad T_{\mathbb S}}$.
This has components referred to as \emph{quotient maps} that are
characterised by being the unique 
morphisms~$\qt^\mathbb{S}_X: TX\rightarrow T_\mathbb{S}X$ for which the
diagram
\begin{equation}
\label{eqn:quotient-map}
\begin{minipage}{\textwidth}
\xymatrix@C=4pc{
  TTX
  \ar[r]^-{T(\qt^\mathbb{S}_X)}
  \ar[d]_-{\mu_X}
  &
  TT_{\mathbb S}X
  \ar[d]^-{\tau^{\mathbb S}_X}
  \\
  TX
  \ar[r]^-{
\qt^\mathbb{S}_X}
  &
  T_{\mathbb S}X
  \\
  X
  \ar[u]^-{\eta_X}
  \ar[ru]_-{\eta^{\mathbb S}_X}
}
\end{minipage}
\end{equation}
commutes.
\hide{
The family of maps $\setof{\tau^\mathbb{S}_X}_{X\in\cat C}$ and
$\setof{\qt^\mathbb{S}_X}_{X\in\cat C}$ are natural in $X$.}
In this situation, we have a form of \emph{strong completeness} stating that
an equation is satisfied in all models if and only if it is satisfied in a
freely generated one, if and only if it is identified by the quotient map.
\begin{theorem}[Internal completeness]
\label{thm:int-comp}
For a MES $\mathbb S = (\ul{\cat C},\monad T,\Ax)$ admitting free algebras, the
following are equivalent. 
\medskip
\begin{enumerate}[\em(1)]
\item\label{thm:int-comp-1}
$\catalg{\mathbb S} \models \eqnz u v:C\rightarrow TA$.\\ 

\item\label{thm:int-comp-2}
$(T_\mathbb{S} A,\tau^\mathbb{S}_A) \models \eqnz u v:C\rightarrow TA$.\\

\item\label{thm:int-comp-3}
$\qt^\mathbb{S}_A\comp u \;=\; \qt^\mathbb{S}_A\comp v:
C\rightarrow T_\mathbb{S} A$.\\
\end{enumerate}
\end{theorem}
\proof 
See Appendix~\ref{IntCompApp}.
\qed 

\subsection{Free Algebras}
%
We now establish a general setting in which to apply the internal
completeness theorem.  Indeed, we give conditions under which MESs admit
free algebras and provide an inductive construction of quotient
maps~(see~\citet{FioreHurTCS} for details).

\begin{definition}
  Let $\ul{\cat C}$ be a symmetric monoidal closed category.  An object $A$ in
  $\cat C$ is respectively said to be \emph{compact} and \emph{projective} if
  the endofunctor ${[A,-]}$ on $\cat C$ respectively preserves colimits of 
  \mbox{$\omega$-chains} and epimorphisms.  
\end{definition}

\begin{definition}
  A MES ${\mathbb S =(\ul{\cat C},\monad T,\Ax)}$ is called \emph{finitary} if
  the category $\cat C$ is cocomplete, the endofunctor $T$ on $\cat C$ is
  \mbox{$\omega$-cocontinuous}, and the arity $A$ of each equation 
  $\eqnz u v:C\rightarrow TA$ in $\Ax$ is compact.  Such a MES is
  called \emph{inductive} if furthermore the endofunctor~$T$ preserves 
  epimorphisms and the arity $A$ of each equation 
  $\eqnz u v:C\rightarrow TA$ in $\Ax$ is projective.
\end{definition}

For a finitary MES $\mathbb S = (\ul{\cat C},\monad T,\Ax)$ we have the
following situation: 
\[
\xymatrix@C=3pc{
\catalg{\mathbb{S}}\
\ar@<-5pt>@{^(->}[r]_-{J} 
\ar@{}@<2pt>[r]|-{\bot}
& \catAlg{\cat C}{\monad T}
\ar@<2pt>[d]^(.5){U_{\monad T}} 
\ar@{}@<-3pt>[d]|-{\dashv}
\ar@/_10pt/[l]_-{K}
\\
& \cat C  
\ar@/^10pt/[u]^-{}
}
\]
For each object $X\in\cat C$, since $(TX,\mu_X)$ is a free Eilenberg-Moore
algebra on $X$, the free \mh{\mathbb S} algebra
$(T_\mathbb{S}X,\tau^\mathbb{S}_X)$ on $X$ is given by the free
\mh{\mathbb S} algebra $K(TX,\mu_X)$ over the Eilenberg-Moore algebra
$(TX,\mu_X)$.  Satisfying the commutative diagram~(\ref{eqn:quotient-map}),
the universal homomorphism 
$(TX,\mu_X)\rightarrow (T_\mathbb{S}X,\tau^\mathbb{S}_X)$ induced by the
adjunction $K \dashv J$ yields the quotient 
map~$\qt^\mathbb{S}_X:TX \rightarrow T_{\mathbb S}X$.

In the case of inductive MESs, the quotient maps $\qt^\mathbb{S}_X$ are
constructed as follows:
\begin{equation}
\label{eqn:qt-con}
\begin{minipage}{.9\textwidth}
\hspace*{-1.2pc}
\xymatrix@C=1.7pc{
\save[]+<-0.2pc,-0.5pc>
*{\mbox{\small${\forall\,(\eqnz u v:C \rightarrow TA) \in \Ax}$}}
\restore
&
T(TX)
\ar[rd]^-{p_0}
\ar@{->>}[r]^-{T(q_0)}
\ar[d]_(.45){\mu_X} 
\ar@{}[rrd]|-{\textsf{po}}
& 
T(TX)_1 
\ar[dr]^-{p_1} 
\ar@{->>}[r]^{T(q_1)} 
\ar@{}[rrd]|-{\textsf{po}}
&
T(TX)_2
\ar@{->>}[r]^{T(q_2)}  
\ar[dr]^-{p_2} 
& 
T(TX)_3 
\ar@{}[r]|-{\cdots} 
& 
T(T_\mathbb{S}X) 
\ar@{-->}[d]^-{\tau^\mathbb{S}_X}  
\\
\save[]+<-2pc,0pc>
*+<6pt,6pt>{[A,TX]\tensor C}
\ar@<3pt>[r]^-{\algstr{\llrrbrk{u}}{(TX,\mu_X)}}
\ar@<-3pt>[r]_-{\algstr{\llrrbrk{v}}{(TX,\mu_X)}}
\ar@{}@<10pt>[r]^(0)*{\vdots}
\ar@{}@<-2pt>[r]_(0)*{\vdots}
\restore
{\mbox{}\hspace*{6pc}\mbox{}}
&
TX
\ar@{->>}@<-4pt>@/_20pt/[rrrr]|-{\ \qt^\mathbb{S}_X\ }
\ar@{->>}[r]^-{q_0} 
\ar@{}[r]_-{\s{coeq}} 
& 
(TX)_1 
\ar@{->>}[r]^-{q_1} 
& 
(TX)_2 
\ar@{->>}[r]^-{q_2} 
& 
(TX)_3 
\ar@{}[r]|-{\cdots} 
\ar@{}@<-7pt>[r]_(1){\textsf{colim}} 
& 
T_\mathbb{S}X 
}
\end{minipage}
\end{equation}
where $q_0$ is the universal map that jointly coequalizes every pair
$\algstr{\llrrbrk u}{(TX,\mu_X)}$ and $\algstr{\llrrbrk v}{(TX,\mu_X)}$ with  
$(\eqnz u v: C\rightarrow TA)\in \Ax$ and where, for all $n\geq1$, the
cospans
\[
(TX)_n 
\xymatrix{\ar@{->>}[r]^-{q_n}&}
(TX)_{n+1}
\xymatrix{&\ar[l]_-{p_n}
T(TX)_n}
\] 
are pushouts of the spans
\[
(TX)_n
\xymatrix{&\ar[l]_-{p_{n-1}}} 
T(TX)_{n-1}
\xymatrix@C=35pt{\ar@{->>}[r]^-{T(q_{n-1})}&} 
T(TX)_n
\] 
for $(TX)_0=TX$.

Moreover, when the strong monad~$\monad T$ arises from a left adjoint to a
forgetful functor~$\catalg F\rightarrow\cat C$, for $F$ a strong endofunctor
that preserves colimits of \mbox{$\omega$}-chains and epimorphisms, the
construction of the quotient maps $\qt^\mathbb{S}_X$ simplifies as follows:
\begin{equation}
\label{eqn:qt-con-simp}
\begin{minipage}{.9\textwidth}
\hspace*{-1.2pc}
\xymatrix@C=1.7pc{
\save[]+<-0.2pc,-0.5pc>
*{\mbox{\small${\forall\,(\eqnz u v:C \rightarrow TA) \in \Ax}$}}
\restore
&
F(TX)
\ar[rd]^-{p_0}
\ar@{->>}[r]^-{F(q_0)}
\ar[d]_(.45){\wh{\mu}_X} 
\ar@{}[rrd]|-{\textsf{po}}
& 
F(TX)_1 
\ar[dr]^-{p_1} 
\ar@{->>}[r]^{F(q_1)} 
\ar@{}[rrd]|-{\textsf{po}}
&
F(TX)_2
\ar@{->>}[r]^{F(q_2)}  
\ar[dr]^-{p_2} 
& 
F(TX)_3 
\ar@{}[r]|-{\cdots} 
& 
F(T_\mathbb{S}X) 
\ar@{-->}[d]^-{\wh{\tau}^\mathbb{S}_X}  
\\
\save[]+<-2pc,0pc>
*+<6pt,6pt>{[A,TX]\tensor C}
\ar@<3pt>[r]^-{\algstr{\llrrbrk{u}}{(TX,\mu_X)}}
\ar@<-3pt>[r]_-{\algstr{\llrrbrk{v}}{(TX,\mu_X)}}
\ar@{}@<10pt>[r]^(0)*{\vdots}
\ar@{}@<-2pt>[r]_(0)*{\vdots}
\restore
{\mbox{}\hspace*{6pc}\mbox{}}
&
TX
\ar@{->>}@<-4pt>@/_20pt/[rrrr]|-{\ \qt^\mathbb{S}_X\ }
\ar@{->>}[r]^-{q_0} 
\ar@{}[r]_-{\s{coeq}} 
& 
(TX)_1 
\ar@{->>}[r]^-{q_1} 
& 
(TX)_2 
\ar@{->>}[r]^-{q_2} 
& 
(TX)_3 
\ar@{}[r]|-{\cdots} 
\ar@{}@<-7pt>[r]_(1){\textsf{colim}} 
& 
T_\mathbb{S}X 
}
\end{minipage}
\end{equation}
where $(TX,\wh{\mu}_X)$ and $(T_\mathbb{S}X,\wh{\tau}^\mathbb{S}_X)$ are 
the \mh{F} algebras respectively corresponding to 
the  Eilenberg-Moore algebras
$(TX,\mu_X)$ and $(T_\mathbb{S}X,\tau^\mathbb{S}_X)$ for
the monad $\monad T$.
(Explicit calculations of this construction feature in
Sections~\ref{AlgThCompleteness} and~\ref{NominalCompleteness}.)

\section*{\ul{\large Part~II. Methodology
}}

In view of the mathematical development of Part~I, we advocate the following
methodology for synthesising \mbox{mono-sorted} equational logics.
\medskip
\begin{enumerate}[(1)]
  \item
    {\em Select a symmetric monoidal closed category~$\ul{\cat C}$ as universe
    of discourse and consider within it a syntactic notion of signature such
    that every signature~$\Sigma$ gives rise to a strong monad~$\monad
    T_\Sigma$ on $\ul{\cat C}$.}

\medskip\noindent
    The universe of discourse should be carefully chosen to consist of
    mathematical objects with enough internal structure to allow for the
    algebraic realisation of the syntactic constructs that one is 
    modelling. 

\medskip\noindent
    We do not insist on an a~priori prescription for the definition of
    signature, but rather consider it as being domain specific.  Of course,
    standard notions of signature~(\eg~as 
    in 
    enriched algebraic theories---see~\cite{KellyPower93},
    and~\citet{Robinson02}) may be considered.  However, one may need to go
    beyond them---see~\citet{Fiore08} and~\citet{FioreHurSOEqLog}.

\medskip
  \item
    {\em Select a class of coarity-arity pairs~$(C,A)$ of objects of
    $\cat C$ and give a syntactic description of Kleisli 
    maps~$C\rightarrow T_\Sigma A$.  This yields a syntactic notion of
    equational presentation with an associated model theory arising from that
    of MESs.}

\medskip\noindent
    We are ultimately interested in constructing free algebras for equational
    presentations.  In the context of finitary algebraic theories, it is thus
    appropriate to consider a cocomplete universe of discourse together with
    signatures for which the associated monad preserves colimits of
    \mbox{$\omega$-chains} and epimorphisms, and arities that are compact and
    projective; so that the induced MESs are inductive.

\medskip
  \item
    {\em Synthesise a deductive system for equational reasoning on syntactic
    terms with rules arising as syntactic counterparts of the EML~rules
    associated to the MES.}

\medskip\noindent
    The analysis of the rule~$\s{Subst}$ will typically involve the
    consideration of a syntactic substitution operation corresponding to
    Kleisli composition.

\medskip
  \item
    {\em Analyse the inductive construction of free algebras and obtain an
    intermediate deductive system characterising the equivalence induced by
    the quotient maps.  Embed the intermediate deductive system within the
    synthesised equational logic and conclude the completeness of the latter
    as a consequence of the internal completeness result.}

\medskip\noindent
    In practise, we have found that the intermediate deductive system is not
    only easily embeddable in the synthesised equational logic but that it
    moreover allows one to distil a rewriting-style deduction system that
    provides a sound and complete computational treatment of derivability.
\end{enumerate}

\medskip
The resulting equational logics are thus synthesised from algebraic
metatheories by means of first principles.  Two sample applications of this
methodology follow.

\section{Synthetic Equational Logic}

\subsection{MESs for algebraic theories}

Recall that an \emph{algebraic theory}~$\mathbb{T}=(\Sigma,E)$ is given by a
signature~$\Sigma$, consisting of a set of operators~$O$ and an arity function
$\U-:O_\Sigma\rightarrow\Nat$, together with a set of equations~$E$.
Algebraic theories may be encoded as MESs as follows.

The signature~$\Sigma$ induces the 
endofunctor~$F_\Sigma(X) = \coprod_{\opr\in O} X^{\U \opr}$ on $\Set$, for
which the category of \mbox{$\Sigma$-algebras}, $\catalg{\Sigma}$, and the
category of \mbox{$F_\Sigma$-algebras}, $\catalg{F_\Sigma}$, are isomorphic.
The forgetful functor~$\catalg{F_\Sigma} \rightarrow \Set$ is monadic and the
induced \emph{term monad}~$\monad{T}_\Sigma =
(T_\Sigma,\eta^\Sigma,\mu^\Sigma)$ is given syntactically.  For a set of
variables~$V$, the set~$T_\Sigma(V)$ consists of terms built up from the
variables in~$V$ and the operators in $O$.

The endofunctor~$F_\Sigma$ has a canonical strength~$\st:U\times F_\Sigma(V)
\rightarrow F_\Sigma(U\times V)$ mapping a
pair~$(u,\iota_{\opr}(v_1,\ldots,v_{\U{\opr}}))$ to
$\iota_{\opr}((u,v_1),\ldots,(u,v_{\U{\opr}}))$, where we use the
notation~$\iota$ for coproduct injections.
The induced strength on the monad~$\monad{T}_\Sigma$, 
$\wh{\st}:U\times T_\Sigma(V) \rightarrow T_\Sigma(U\times V)$, 
maps a pair~$(u,t)$ to the term~$t\subst{v\mapsto (u,v)}_{v\in V}$
obtained by simultaneously substituting $(u,v)$ for each 
variable~$v \in V$ in the term~$t$.

By definition, each equation~$(\eqn V  l r)$ in $E$ is given by a pair of
terms~${l,r \in T_\Sigma(V)}$, or equivalently, by a parallel pair of Kleisli
maps~$\toTES{l},\toTES{r}: 1 \rightarrow T_\Sigma (V)$.  Thus, one can encode
the algebraic theory~$\mathbb T$ as the MES
$\toTES{\mathbb{T}}=(\Set,\monad{T}_\Sigma,\toTES{E})$ with the set of
equations $\toTES{E}$ given by
$\setof{ \eqnz{\toTES{l}}{\toTES{r}}:1\rightarrow T_\Sigma V 
         \suchthat
         (\eqn V l r) \in E}$.
The MES $\toTES{\mathbb T}$ is inductive.

\subsection{Model theory}

A \mh{\toTES{\mathbb T}} algebra is an Eilenberg-Moore 
algebra~$(X,s:T_\Sigma X\rightarrow X)$ for $\monad T_\Sigma$ 
such that the diagram 
\[
\!\!
\xymatrix@C=60pt@R=0pt{
X^V\times 1 
\ar@<3pt>[r]^-{X^V\times \toTES{t_1}} 
\ar@<-3pt>[r]_-{X^V\times \toTES{t_2}} 
& 
}
\!\!\!
\xymatrix@C=45pt@R=0pt{
X^V\times T_\Sigma V
\ar[r]^-{\wh{\st}_{X^V,V}}
&
T_\Sigma\big(X^V\times V\big) 
}
\!\!\!
\xymatrix@C=35pt@R=0pt{
\ar[r]^-{T_\Sigma(\epsilon^V_X)} 
& 
T_\Sigma X 
}
\!\!\!
\xymatrix@C=15pt@R=0pt{
\ar[r]^-{s} 
& 
X
}
\]
commutes for every equation $(\eqn V {t_1}{t_2})$ in $E$; that is, such
that
\begin{equation}
\label{eqn:alg-thy-tes-cond}
\xymatrix@C=2.5pc{
  1
  \ar@<3pt>[r]^-{\toTES{t_1}}
  \ar@<-3pt>[r]_-{\toTES{t_2}}
  &
  T_\Sigma V
  \ar[r]^-{T_\Sigma (v)}
  &
  T_\Sigma X
  \ar[r]^-{s}
  &
  X
}
\end{equation}
commutes for all functions $v:V\rightarrow X$.

Write $(X,\llrrbrk{-})$ for the Eilenberg-Moore algebra of the 
monad~$\monad T_\Sigma$ corresponding to the 
\mh{\Sigma} algebra~$(X,\setofz{\llrrbrk{\opr}}_{\opr\in\Sigma})$ via the
isomorphism~$\catalg{\Sigma}\iso\catAlg{\cat C}{\monad T_\Sigma}$.  
We have that the Eilenberg-Moore algebra~$(X,\llrrbrk{-})$
satisfies~(\ref{eqn:alg-thy-tes-cond}) if and only if the 
\mh{\Sigma} algebra~$(X,\setofz{\llrrbrk{\opr}}_{\opr\in\Sigma})$ satisfies
the equation~$(\eqn V {t_1} {t_2})$.  It follows thus that
$\catalg{\toTES{\mathbb T}}$ is isomorphic to the 
category~$\catalg{\mathbb T}$ of algebras for the algebraic 
theory~$\mathbb T$.

\subsection{EML for algebraic theories}

The EML associated to the MES of an algebraic theory~$\mathbb T = (\Sigma,
E)$ has judgements of the form\\[-1.5mm]
\[
\eqn{\ul E}{f}{g}:U\rightarrow T_\Sigma V
\]\\[-2.5mm]
with inference rules~$\s{Ref}$, $\s{Sym}$, $\s{Trans}$, $\s{Axiom}$,
$\s{Subst}_\amalg$, $\s{Ext}$, and~$\s{Local}_1$ (see
Section~\ref{EquationalMetalogic}).  
The rules~$\s{Ext}$ and~$\s{Local}_1$ are however redundant.  Indeed, the
subsystem~EML$_1$ with inference rules~$\s{Ref}$, $\s{Sym}$, $\s{Trans}$,
$\s{Axiom}$, and $\s{Subst}_\amalg$ restricted to judgements of the
form\\[-1mm]
\[
\eqn{\ul E}{u}{v}:1\rightarrow T_\Sigma V
\]\\[-3mm]
is such that\\[-1mm]
\begin{quote}
\begin{tabular}{rl}
& $\eqn{\ul E}{f}{g}:U\rightarrow T_\Sigma V$ is derivable in EML
\\
iff &
\\
& $\eqn{\ul E}{f\subst{\ul i}}{g\subst{\ul i}}:1\rightarrow T_\Sigma V$ is
derivable in EML$_1$ for all $i\in U$.
\end{tabular}
\end{quote}

\subsection{Synthetic Equational Logic}

A Synthetic Equational Logic~(SEL) for algebraic 
theories~$\mathbb T = (\Sigma, E)$ directly arises as the syntactic
counterpart of EML$_1$.  SEL has judgements 
\[
V \vdash_E s \equiv t 
\quad
(s,t\in T_\Sigma V)
\]
and consists of the following rules:\\
\begin{equation*}
\label{eqn:tel-for-algthy}
\begin{minipage}{.9\textwidth}
\mbox{}\hfill$
\begin{array}{c}
\myproof{
\AxiomC{$\phantom{V}$}
\LeftLabel{$\s{Ref}$}
\RightLabel{}
\UnaryInfC{$V \vdash_E t \equiv t$}
}
\qquad
\myproof{
\AxiomC{$V\vdash_E {t}\equiv{t'}$}
\LeftLabel{$\s{Sym}$}
\UnaryInfC{${V}\vdash_E{t'}\equiv{t}$}
}
\qquad
\myproof{
\AxiomC{${V}\vdash_E{t}\equiv{t'}\quad {V}\vdash_E{t'}\equiv{t''}$}
\LeftLabel{$\s{Trans}$}
\UnaryInfC{${V}\vdash_E{t}\equiv{t''}$}
}
\\[6mm]
\myproof{
\AxiomC{$({V}\vdash{l}\equiv{r}) \in E$}
\LeftLabel{$\s{Axiom}$}
\RightLabel{}
\UnaryInfC{${V}\vdash_E{l}\equiv{r}$}
}
\\[6mm]
\myproof{
\AxiomC{$
{U}\vdash_E{t}\equiv{t'}
\qquad
{V}\vdash_E{s_u}\equiv{s'_u}\ (u\in U)
$}
\LeftLabel{$\s{Subst}$}
\UnaryInfC{$
{V} \vdash_E
{t\subst{u\mapsto s_u}_{u\in U} \,}
\equiv
{\, t'\subst{u\mapsto s'_u}_{u\in U}}
$}
}
\end{array}
$\hfill\mbox{}
\end{minipage}
\end{equation*}\\[1mm]
In the rule~$\s{Subst}$, the term $t\subst{u\mapsto s_u}_{u\in U}$ is
obtained by simultaneously substituting the terms~$s_u$ for the variables
$u\in U$ in the term $t$.

\subsection{Soundness}

Note that SEL subsumes the usual presentation of Equational Logic, where
the substitution rule is restricted to families~${s_u}={s'_u}~(u\in U)$
and a congruence rule for operators is added.
Furthermore, since $V\vdash_E s\equiv t$ is derivable in SEL iff 
$\eqn{\ul E}{\ul s}{\ul t}:1\rightarrow T_\Sigma V$ is derivable in
EML$_1$ iff $\eqn{\ul E}{\ul s}{\ul t}:1\rightarrow T_\Sigma V$ is
derivable in EML, the well-known soundness of SEL follows from the
soundness of EML.

\subsection{Completeness}
\label{AlgThCompleteness}

We proceed to show how the internal completeness 
theorem 
and the construction of free algebras for 
inductive~MESs
~(see Section~\ref{InternalCompletenessSection}) 
lead to equational derivability and bidirectional rewriting completeness
results.

Consider the construction~(\ref{eqn:qt-con-simp}) for the 
MES~$\toTES{\mathbb{T}}$.
The map~$q_0:T_\Sigma X \epirightarrow (T_\Sigma X)_1$ is the universal map in
$\Set$ that coequalizes every pair 
$\llrrbrk{l},\llrrbrk{r}:(T_\Sigma X)^V\rightarrow T_\Sigma X$ for all
${(\eqn{V}{l}{r})\in E}$, where $\llrrbrk{t}$ maps 
$s \in (T_\Sigma X)^V$ to $t\subst{v \mapsto s_v}_{v\in V}$.
It follows that the set $(T_\Sigma X)_1$ is given by the 
quotient~$T_\Sigma X/_{\approx_1}$ of $T_\Sigma X$ under the equivalence
relation~$\approx_1$ generated by the rule:\\
\[
\myproof{
  \AxiomC{$(\eqn{V}{l}{r})\in E$}
  \RightLabel{$\big(s \in (T_\Sigma X)^V\big)$}
  \UnaryInfC{$
    l\subst{v \mapsto s_v}_{v\in V} 
    \approx_1
    r\subst{v \mapsto s_v}_{v\in V}
    $}
}
\]\\
The map~$q_0$ sends a term~$t\in T_\Sigma X$ to 
its equivalence class $[t]_{\approx_1}\in T_\Sigma X/_{\approx_1}$,
and the map~$p_0$ sends 
$\iota_{\opr} (t_1,\ldots,t_{\U\opr}) \in F_\Sigma(T_\Sigma X)$
to $[\opr(t_1,\ldots,t_{\U\opr})]_{\approx_1} \in T_\Sigma X/_{\approx_1}$.

Recall that a pushout 
\[
\xymatrix{\ar@{}[rrd]|-{\mathsf{po}}
  A
  \ar@{->>}[r]^-{e}
  \ar[rd]_-{f}
  &
  B
  \ar[rd]^-{f'}
  \\
  &
  C
  \ar@{->>}[r]^-{e'}
  &
  C/_{\approx}
}
\]
of a surjective map~$e:A\epirightarrow B$ and a map~$f:A\rightarrow C$ in
$\Set$ can be constructed as the 
cospan~$e':C\epirightarrow C/_{\approx} \leftarrow B:f'$, where 
$C/_{\approx}$ is the quotient of $C$ under the equivalence
relation~$\approx$ generated by setting~$f(a) \approx f(a')$ for all 
$a,a'\in A$ such that $e(a) = e(a')$ in $B$, and where the surjective
map~$e':C\epirightarrow C/_{\approx}$ sends an element $c$ to its equivalence
class~$[c]_{\approx}$ and the map~$f':B \rightarrow C/_{\approx}$ sends an
element $b$ to~$e'(f(a))$ for $a\in A$ such that $e(a)=b$.

Using this construction, an inductive analysis of the maps~$q_n$ for~$n
\ge 1$ shows that the sets $(T_\Sigma X)_n$ for~$n \ge 2$ are given as the
quotients~$T_\Sigma X/_{\approx_n}$ of $T_\Sigma X$ under the equivalence
relations~$\approx_n$ inductively generated by the following rules:\\
\[
\myproof{
  \AxiomC{$
    s \approx_{n-1} s'
    $}  
  \UnaryInfC{$
    s \approx_n s'
    $}
}
\qquad\qquad
\myproof{
  \AxiomC{$
    s_1 \approx_{n-1} s'_1,\ \ldots,\  
    s_{\U\opr} \approx_{n-1} s'_{\U\opr}
    $}  
  \RightLabel{$
    (\opr \in \Sigma)
    $}
  \UnaryInfC{$
    \opr(s_1,\ldots,s_{\U\opr}) 
    \approx_n
    \opr(s'_1,\ldots,s'_{\U\opr}) 
    $}
}
\]\\
The maps~$q_n$ for $n\ge 1$ send $[t]_{\approx_n}\in T_\Sigma X/_{\approx_n}$
to $[t]_{\approx_{n+1}} \in T_\Sigma X/_{\approx_{n+1}}$, and the maps $p_n$
for $n\ge 1$ send $\iota_{\opr}
([t_1]_{\approx_n},\ldots,[t_{\U\opr}]_{\approx_n}) \in F_\Sigma(T_\Sigma
X/_{\approx_n})$ to $[\opr(t_1,\ldots,t_{\U\opr})]_{\approx_{n+1}} \in
T_\Sigma X/_{\approx_{n+1}}$.

By taking the colimit of the chain of quotients 
${\setof{q_n:T_\Sigma X/_{\approx_n} 
  \epirightarrow 
  T_\Sigma X/_{\approx_{n+1}}}_{n\ge 0}}$, 
the set $T_\mathbb{S} X$ is given by the quotient
$T_\Sigma X/_{\approx_E}$ of $T_\Sigma X$
under the relation $\approx_E$ generated by the following rules:
\begin{equation}
\label{eqn:rw-for-algthy}
\begin{minipage}{.9\textwidth}
\mbox{}\hfill$
\begin{array}{c}
\myproof{
  \AxiomC{$\phantom{t}$}
  \LeftLabel{$\s{Ref}$}
  \RightLabel{}
  \UnaryInfC{$
    s \approx_E s
    $}
}
\qquad\qquad
\myproof{
  \AxiomC{$
    s \approx_E s'
    $}
  \LeftLabel{$\s{Sym}$}
  \UnaryInfC{$
    s' \approx_E s
    $}
}
\qquad\qquad
\myproof{
  \AxiomC{$
    s \approx_E s' \quad s' \approx_E s''
    $}
  \LeftLabel{$\s{Trans}$}
  \UnaryInfC{$
    s \approx_E s''
    $}
}
\\[6mm]
\!\!\!\!\!\!
\myproof{
  \AxiomC{$(\eqn{V}{l}{r})\in E$}
  \LeftLabel{$\s{Inst}$}
  \RightLabel{$\big(s\in (T_\Sigma X)^V\big)$}
  \UnaryInfC{$
    l\subst{v \mapsto s_v}_{v\in V} 
    \approx_E
    r\subst{v \mapsto s_v}_{v\in V}
    $}
}
\\[6mm]
\myproof{
  \AxiomC{$
    s_1 \approx_E s'_1,\ \ldots,\  
    s_{\U\opr} \approx_E s'_{\U\opr}
    $}  
  \LeftLabel{$\s{Cong}$}
  \RightLabel{$
    (\opr \in \Sigma)
    $}
  \UnaryInfC{$
    \opr(s_1,\ldots,s_{\U\opr}) 
    \approx_E
    \opr(s'_1,\ldots,s'_{\U\opr}) 
    $}
}
\end{array}
$\hfill\mbox{}
\end{minipage}
\end{equation}
The quotient map~$\qt^\mathbb{S}_X:T_\Sigma X \epirightarrow T_\mathbb{S}X$
sends a term $t$ to its equivalence class $[t]_{\approx_E}$.

The rules~$\s{Inst}$ and $\s{Cong}$ for the relation $\approx_E$ can be
merged into a single rule to yield a rewriting-style deduction system.
Indeed, by an induction on the depth of proof trees, one shows that the
relation~$\approx_E$ on $T_\Sigma X$ coincides with the equivalence
relation~$\approx^\s{R}_E$ on $T_\Sigma X$ generated by the rewriting-style
rule\\[-2mm]
\[
\mbf{C}[l\subst{v \mapsto s_v}_{v\in V}]
\;\approx^\s{R}_E\;
\mbf{C}[r\subst{v \mapsto s_v}_{v\in V}]
\]\\[-2mm]
for $(\eqn{V}{l}{r})\in E$, $s\in (T_\Sigma X)^V$, and $\mbf{C}[-]$ a context
with one hole and possibly with variables from $X$.

From the internal completeness of the MES~$\toTES{\mathbb T}$, we have the
soundness and completeness of equational reasoning by bidirectional
rewriting: 
\[
\begin{array}{cl}
&
\hspace{-17.5mm}
\catalg{\mathbb{T}} \models (\eqn{V}{s}{t})
\\[1mm]
\Longleftrightarrow&
\catalg{\toTES{\mathbb{T}}} 
  \models (\eqnz{\toTES{s}}{\toTES{t}}:1 \rightarrow T_\Sigma V)
\\[1mm]
\Longleftrightarrow&
\qt^\mathbb{S}_V \comp \toTES{s} = \qt^\mathbb{S}_V \comp \toTES{t}
: 1\rightarrow T_\mathbb{S} V
\\[1mm]
\Longleftrightarrow&
[s]_{\approx_E} = [t]_{\approx_E} \quad\text{in } T_\Sigma V/_{\approx_E}
\\[1mm]
\Longleftrightarrow&
s \approx_E t \quad\text{in } T_\Sigma V
\\[1mm]
\Longleftrightarrow&
s \approx^\s{R}_E t \quad\text{in } T_\Sigma V
\end{array}
\]

Finally, also SEL is complete; as a proof of $s \approx_E t$ for 
${s,t\in T_\Sigma V}$ constructed by the rules
in~(\ref{eqn:rw-for-algthy}) can be turned into a proof of
${V}\vdash_E{s}\equiv{t}$ in SEL.


\section{Synthetic Nominal Equational Logic}
\label{sec:syn-nom-eqn-logic}

This section provides a novel application of our theory and methodology
for synthesising equational logics geared to the development of a
deductive system for reasoning about algebraic structure with name-binding
operators.

We consider a class of MESs, referred to as 
\emph{Nominal Equational Systems~(NESs)}, based on the category~$\Nom$ of
nominal sets~\citep[Section~6]{GabbayPitts01} (or equivalently the
Schanuel topos~\citep[Section~III.9]{MacLaneMoerdijk92}).  
These we subsequently present in syntactic form to yield \emph{Nominal
Equational Presentations~(NEPs)}.  The model theory of NEPs is of course
derived from that of NESs.

An equational logic, called 
\emph{Synthetic Nominal Equational Logic~(SNEL)}, for NEPs is derived
from the EML associated to NESs.  This is guaranteed to be sound by
construction.  Completeness is derived from the internal completeness
theorem by an analysis of the inductive construction of free algebras in
terms of equational derivability.  This approach yields two completeness
results: the rewriting completeness of an induced notion of 
\emph{Synthetic Nominal Rewriting~(SNR)} and the derivability completeness of
SNEL.

A brief discussion of related work is included.

\subsection{Nominal sets}
\label{sec:nom-set}

For a fixed countably infinite set $\Atom$ of atoms, the group $\Perm$
of finite permutations of atoms consists of the bijections on $\Atom$
that fix all but finitely many elements of $\Atom$.
A \emph{{\mh\Perm action}} 
${X=(\,\U X,\permact)}$ consists of a set $\U X$
equipped with a function
${(-)\permact(=)} : {\Perm\times \U X\rightarrow \U X}$
satisfying ${\id_{\Atom} \permact x = x}$ and
${\pi'\permact(\pi\permact x) = (\pi'\pi)\permact x}$ for all ${x\in \U X}$ 
and ${\pi,\pi'\in\Perm}$.
{\mh\Perm actions} form a category 
with morphisms ${X\rightarrow Y}$ given
by \emph{equivariant functions}; that is, 
functions ${f:\U X\rightarrow \U Y}$ such that 
${f(\pi \permact x) = \pi\permact (fx)}$
for all ${\pi\in\Perm}$ and ${x\in \U X}$.

For a {\mh\Perm action} $X$, 
a finite subset $S$ of $\Atom$ is said to
\emph{support} $x\in X$ if for all atoms~$a,a'\not\in S$, we have 
that $\atmtrans{a}{a'} \permact x = x$, where the \emph{transposition}
$\atmtrans{a}{a'}$ is the bijection that swaps $a$ and $a'$, and fixes
all other atoms. 
A \emph{nominal set} 
is a {\mh\Perm action} in which every element has finite support.
As an example, the set of atoms $\Atom$ becomes the
\emph{nominal set of atoms} 
$\NomAtom$ when equipped with the evaluation
action~$\pi\permact a = \pi(a)$.
The category $\Nom$ is the full subcategory of the category of
\mh\Perm actions consisting of nominal sets.

The supports of an element of a nominal set are closed under 
intersection, and we write $\supp X x$, or simply $\suppz x$,
for the intersection of the supports of
$x$ in the nominal set~$X$.
For instance, 
we have that $\supp{\NomAtom}{a}=\setof{a}$.
For 
elements $x$ and $y$ of two, possibly distinct, nominal sets $X$ and $Y$, we
write $x\freshfor y$ whenever $\supp X x$ and $\supp Y y$ are disjoint. 
Thus, for $a\in\NomAtom$ and 
$x\in X$, 
$a\freshfor x$ stands for $a\not\in\supp X x$; that is, 
\emph{$a$ is fresh for $x$}.

The category $\Nom$ is complete and cocomplete. 
In particular, for a possibly infinite family of nominal sets 
$\setof{X_i}_{i\in I}$,
the coproduct $\coprod_{i\in I} X_i$ 
is given by  $\U{\coprod_{i\in I} X_i}= \coprod_{i\in I} \U{X_i}$
with action ${\pi\permact\inj i(x) = \inj i (\pi\permact x)}$, where we
use the notation $\iota$ for coproduct injections;
the product
$\prod_{i\in I} X_i$, for a finite set $I$,
is given by 
$\U{\prod_{i\in I} X_i}=\prod_{i\in I} \U{X_i}$ 
with action
$\pi\permact\setofz{x_i}_{i\in I} = \setofz{\pi\permact x_i}_{i\in I}$.
As usual, we write $X^n$ for the \mbox{$n$-fold} product~$X \times \cdots
\times X$.

Further, $\Nom$ carries a symmetric monoidal closed structure
$(1,\septen, \sephom{-,=})$.
The \emph{unit}~$1$ is 
the terminal object in $\Nom$
(\ie, the singleton set consisting of the empty tuple
equipped with the unique action).
The \emph{separating tensor} $X\septen Y$ is the 
nominal subset of $X\times Y$ with underlying set given by
$\setof{(x,y)\in \U X\times\U Y \suchthat x \freshfor y}$.
We write $X^{\septen n}$ for the \mbox{$n$-fold} tensor
product~$X\septen\cdots\septen X$.  For instance, $\atmprd{n}$
consists of
$n$-tuples of distinct atoms
equipped with the pointwise action
$\pi\permact (a_1,\ldots,a_n) 
=(\pi\permact a_1,\ldots,\pi\permact a_n)$.
Note that $X^{\septen 0}$ is $1$ for any nominal set X.
Henceforth we write $\atmlst a n$, or simply $\atmlstz a$ when $n$ is
clear from the context, as a shorthand for a tuple $a_1,\ldots,a_n$ of
distinct atoms, and further write $\setofz{\atmlst a n}$ for the set
$\setofz{a_1,\ldots,a_n}$.
For pairs $\atmlsth a n,\atmlsth b n\in\atmprd{n}$ we define the
\emph{multi-transposition}~$\atmtrans{\atmlsth a n}{\atmlsth b n}$ to be a
\emph{fixed} bijection on $\Atom$ such that ${\atmtrans{\atmlsth a
n}{\atmlsth b n}(a_i) = b_i}$ for all $1\leq i \leq n$, and
${\atmtrans{\atmlsth a n}{\atmlsth b n}(c) = c}$ for all ${c \not\in
\setofz{\atmlsth a n} \cup \setofz{\atmlsth b n}}$.

The separating tensor $\septen$ is closed and the associated internal-hom
functor is denoted $\sephom{-,=}$.  
In particular, the internal homs $\atmhom{n}{X}$, for $n \in \Nat$ and
$X\in\Nom$, give rise to a notion of multi-atom~abstraction.  Indeed, the
nominal set $\atmhom n X$ has underlying set given by the
quotient~${\U{\atmprd{n}\times X}/_{\approx_\alpha}}$ under the
\emph{$\alpha$-equivalence relation}~$\approx_\alpha$ defined
as follows:
${(\atmlstz a, x) \approx_\alpha (\atmlstz b, y)}$
if and only if there exists a \emph{fresh} ${\atmlstz c\in \atmprd{n}}$
(\ie,~a tuple $\atmlstz c \in \atmprd{n}$ satisfying the condition 
$\atmlstz c \freshfor \atmlstz a, x, \atmlstz b, y$)
such that 
$\atmtrans{\atmlstz a}{\atmlstz c}\permact x = 
 \atmtrans{\atmlstz b}{\atmlstz c}\permact y$.
We write $\atmabs{\atmlstz a} x$ for the equivalence class 
$[(\atmlstz a,x)]_{\approx_\alpha}$.  The nominal set $\atmhom n X$ has
action 
$\pi\permact\atmabs{\atmlstz a} x
=
\atmabs{ \pi\permact \atmlstz a} 
{\pi\permact x}$ 
Note that $\suppz{\atmabs{\atmlstz a} x}$ is  
$\suppz{x}{\setminus}\setofz{\atmlstz a}$.

\subsection{Nominal Equational Systems}

We specify a class of MESs on $\Nom$, called 
\emph{Nominal Equational Systems~(NESs)}.

Following~\citep{CloustonPitts07} we define a 
\emph{nominal signature}~$\Sigma$ to be a family of nominal
sets~$\setof{\Sigma(n)}_{n\in\Nat}$, each of which consists of the
operators of arity~$n$.

\begin{example}\label{ex:nom-sig-lam}
The nominal signature~$\Sigma_\lambda$ for the untyped
\mbox{$\lambda$-calculus} is given by the nominal sets of operators   
\[
\Sigma_\lambda(0) = \setof{\mbs{V}_a \suchthat a\in\NomAtom} \,, \quad
\Sigma_\lambda(1) = \setof{\mbs{L}_a \suchthat a\in\NomAtom} \,, \quad
\Sigma_\lambda(2) = \setof{\mbs{A}}
\]
with action
\[
\pi \permact \mbs{V}_a = \mbs{V}_{\pi(a)}\ , \quad
\pi \permact \mbs{L}_a = \mbs{L}_{\pi(a)}\ , \quad
\pi \permact \mbs{A} = \mbs{A} \enspace .
\]
\end{example}

To each nominal signature~$\Sigma$,
we associate the strong endofunctor $(F_{\Sigma},\st^\Sigma)$ on
$\ul{\Nom} = {(\Nom,1,\septen, \sephom{-,=})}$ as follows:\\[-1mm]
\[
\begin{array}{rcccc}
F_{\Sigma}(X) &=& 
\multicolumn{3}{l}{\coprod_{n\in\Nat}\; \Sigma(n) \times X^{n}\ ,}
\\[2mm]
\st^\Sigma_{X,Y}
&:& F_\Sigma (X) \septen Y &\rightarrow& F_\Sigma(X\septen Y)
\\[1mm]
&&
\big(\inj n (\opr,x_1,\ldots,x_n),y\big) 
&\mapsto&
\inj n \big(\opr,(x_1,y),\ldots,(x_n,y)\big)
\end{array}
\]\\
for $X,Y\in\Nom$ and $n\in\Nat$, $\opr\in\Sigma(n)$, 
$x_1,\ldots,x_n\in X$, $y\in Y$.
Since $\Nom$ is cocomplete and the functor $F_{\Sigma}$ is \mh{\omega}
cocontinuous, free \mh{F_{\Sigma}} algebras exist.  The carrier of the
free \mh{F_{\Sigma}} algebra $T_\Sigma X$ on $X$ has the following
inductive syntactic description:\\[-1mm]
\begin{equation*}\label{eqn:nom-term}
\begin{array}{rclcl}
t \in T_\Sigma X
& ::= & x
&& (\, x\in X\,)
\\[1mm]
& \mid & \opr\,(t_1,\ldots,t_n)
&\quad& (\,\opr\in\Sigma(n),\; t_1,\ldots,t_n\in T_\Sigma X\,)
\end{array}
\end{equation*}\\
with action 
given by 
$\pi\permact x = \pi \permact_{X} x$
and 
$\pi\permact \opr(t_1,\ldots,t_n) = 
(\pi \permact\opr)(\pi\permact t_1,\ldots,\pi\permact t_n)$.
The associated 
\emph{term monad}~$\monad T_\Sigma = (T_\Sigma,\eta^\Sigma,\mu^\Sigma)$ is
strong, with strength $\wh{\st}^\Sigma$ given as follows:\\[-1mm]
\[
\begin{array}{rcccc}
\wh{\st}^\Sigma_{X,Y}
&:& T_\Sigma (X) \septen Y &\rightarrow& T_\Sigma(X\septen Y)
\\[1mm]
&:&
(t,y)
&\mapsto&
t\subst{x\mapsto (x,y)}_{x\in X}
\end{array}
\]\\
where $t\subst{x\mapsto (x,y)}_{x\in X}$ denotes the term obtained by
simultaneously substituting $(x,y)$ for each $x$ in the term $t$.

\hide{
A \nh{MES} syntax for a NES is given by the strong monad
$\monad T_\Sigma = (T_\Sigma,\eta^\Sigma,\mu^\Sigma,\wh{\st}^\Sigma)$ 
determined by a nominal signature $\Sigma$.
}

\hide{
As equations for NESs,
we only consider equations of the form
$\atmprd n \rightarrow T_\Sigma \coprod_{i=1}^\ell \atmprd{n_i}$
for $n, \ell,n_1,\ldots,n_\ell\in\Nat$.
}
A NES is a MES $(\ul{\Nom},\monad T_\Sigma,\Ax)$ for $\Sigma$ a nominal
signature and $\Ax$ a set of equations with 
arities~$\coprod_{i=1}^\ell \atmprd{n_i}$ for
$\ell,n_1,\ldots,n_\ell\in\Nat$ and coarities $\atmprd n$ for $n\in\Nat$.

\subsection{Nominal Equational Presentations}

We introduce \emph{Nominal Equational Presentations~(NEPs)} as syntactic
counterparts of NESs.

We define a \emph{variable context} as a finite set of variables~$\U{V}$
together with a function~$V:\U{V} \rightarrow \Nat$ assigning a
\emph{valence} to each variable.
We write ${x_1:n_1,\ldots,x_\ell:n_\ell}$ for the variable context
with variables $x_1,\ldots,x_\ell$ respectively of valence
$n_1,\ldots,n_\ell$. 
Variable contexts are syntax for arities, and every such $V$ determines
the arity 
\[
\toTES{V}= \coprod_{x\in\U{V}} \atmprd{V(x)}\;.
\] 
We write $x(\atmlstz{a})$ for the element $\iota_{x}(\atmlstz{a})$ of
$\toTES{V}$; when convenient, we further abbreviate $x()$ as $x$.

For $n\in\Nat$ and a variable context $V$, the following bijection
\begin{equation}\label{NominalKmapBijection}
\begin{array}{l}
\textstyle
\Nom(\atmprd{n}, T_\Sigma\toTES{V})
\\[1mm] 
\textstyle
\quad\ \iso\
\Nom(1,\atmhom n {\, T_\Sigma\toTES{V}})
\\[1mm] 
\textstyle
\quad\ \iso\
\setof{\tau\in \atmhom n {\, T_\Sigma\toTES{V}}
       \;\suchthat\; \suppz{\tau} = \emptyset}
\\[1mm] 
\textstyle
\quad\ =\
\setof{
\atmabs{\atmlstz{a}} t\in 
\atmhom n {\, T_\Sigma\toTES{V}}
\;\suchthat\; 
\suppz{t}\subseteq \setofz{\atmlstz a}}
\end{array}
\end{equation}
shows that a Kleisli map~$\atmprd{n}\rightarrow T_\Sigma\toTES{V}$ is
determined by an \mh{\alpha} equivalence class $\atmabs{\atmlstz{a}} t$
for $\atmlstz a\in\atmprd n$ and 
$t \in (T_\Sigma\toTES{V})_{\atmlstz a} =
\setof{ t\in T_\Sigma\toTES{V} \suchthat \suppz{t}\subseteq
\setofz{\atmlstz a} }$.  
\hide{
Syntactically, the \mh{\alpha} equivalence class $\atmabs{\atmlstz{a}} t$ 
may be described by pairs
\[
(\atmlstz a,t)
\quad\mbox{ for }\;
\atmlstz a \in \atmprd n,\,
t \in T_\Sigma\toTES{V}
\;\mbox{ such that }\;
\suppz{t}\subseteq\setof{\atmlstz a}
\]
where we understand the tuple of distinct atoms $\atmlstz a$ as binding
atoms; and the condition $\suppz{t}\subseteq\setof{\atmlstz a}$
as saying that there are no free atoms in the term $t$.
}
The set $(T_\Sigma\toTES{V})_{\atmlstz a}$ has the following inductive
syntactic description:
\[
\begin{array}{rclcl}
t \in (T_\Sigma\toTES{V})_{\atmlstz a}
& ::= & x(\atmlstz b)
&& (\, x(\atmlstz b) \in \toTES{V}
\text{ such that }\setofz{\atmlstz b}\subseteq\setofz{\atmlstz a} \,)
\\[1mm]
& \mid & \opr(t_1,\ldots,t_n)
&\quad& 
(\,\opr\in\Sigma(n)
\text{ such that }\suppz{\opr}\subseteq\setofz{\atmlstz a} ,
\\
&&&&
\
\text{ and } t_1,\ldots,t_n\in (T_\Sigma\toTES{V})_{\atmlstz a}
\,) \enspace.
\end{array}
\]

Directly motivated by this analysis, we define the notion of \emph{Nominal
Equational Presentation~(NEP)} as follows.  
A \emph{nominal context} $\nomctx{\atmlstz{a}}{V}$ consists of an atom
context~${\atmlstz{a} \in \atmprd{n}}$, for $n\in\Nat$, and a variable
context $V$.
A \emph{nominal term}~$t$, for a nominal signature~$\Sigma$, in a nominal
context~$\nomctx{\atmlstz{a}}{V}$, denoted 
$\nomctx{\atmlstz{a}}{V} \vdash t$, is given by a 
term~$t \in (T_\Sigma\toTES{V})_{\atmlstz a}$; 
a \emph{nominal equation}~$\eqn{\nomctx{\atmlstz a}V}{t}{t'}$ is given by
a pair of nominal terms $t$ and $t'$ in the same nominal 
context~$\nomctx{\atmlstz{a}}{V}$.
A NEP~$\mathbb{T}=(\Sigma,E)$ consists of a nominal signature $\Sigma$
and a set of nominal equations~$E$.

\hide{
We note that the definition of synthetic nominal equational theory 
depends neither on the symmetric monoidal structure $(I,\septen,\sephom{-,=})$
nor on the strength $\wh{\st}^\Sigma$.
As we will see in the next section,
these structures only affect the model theory of 
synthetic nominal equational theories.
}

\begin{example}[continued from Example~\ref{ex:nom-sig-lam}, 
\cf~\citep{GabbayMathijssen07} and\linebreak \citep{CloustonPitts07}]
\label{ex:nom-thy-lam}
The NEP~$\mathbb{T}_\lambda=(\Sigma_\lambda,E_\lambda)$ for
$\alpha\beta\eta$-equivalence of untyped \mbox{$\lambda$-terms} has the
following equations: 
\[
\begin{array}{lrcl}
(\alpha) &
\eqn
{\nomctx{a,b}{\;x:1} &}
{& \mbs{L}_a.\,x(a) \;}
{\; \mbs{L}_b.\,x(b)}
\\[1mm]
(\beta_\kappa) &
\eqn
{\nomctx{a}{\; x:0,y:1 } &}
{& \mbs{A}\big(\mbs{L}_a.\,x\,,\,y(a)\big) \;}
{\; x}
\\[1mm]
(\beta_\mbs{V}) &
\eqn
{\nomctx{a}{\; x:1 } &}
{& \mbs{A}\big(\mbs{L}_a.\,\mbs{V}_a\,,\,x(a)\big) \;}
{\; x(a)}
\\[1mm]
(\beta_\mbs{L}) &
\eqn
{\nomctx{a,b}{\; x:2,y:1}&}
{&\mbs{A}\big(\mbs{L}_a.\,\mbs{L}_b.\,x(a,b)\,,\,y(a)\big)\;}
{\; \mbs{L}_b.\,\mbs{A}\big(\mbs{L}_a.\,x(a,b)\,,\,y(a)\big)}
\\[1mm]
(\beta_\mbs{A}) & 
\terminctx
{\nomctx{a}{\; x:1,y:1,z:1 } &}
{& \mbs{A}\big(\mbs{L}_a.\,\mbs{A}(x(a),y(a))\,,\,z(a)\big) \;}
\\
&&&
\eqnz{\ }
{\; \mbs{A}\big(\mbs{A}\big(\mbs{L}_a.\,x(a),z(a)\big)\,,\,\mbs{A}\big(\mbs{L}_a.\,y(a),z(a)\big)\big)}
\\[1mm]
(\beta_\varepsilon) &
\eqn
{\nomctx{a,b}{\;x:1} &}
{& \mbs{A}\big(\mbs{L}_a.\,x(a),\mbs{V}_{b}\big) \;}
{\; x(b)}
\\[1mm]
(\eta) &
\eqn
{\nomctx{a}{\;x:0} &}
{& \mbs{L}_a.\,\mbs{A}(x,\mbs{V}_a) \;}
{\; x}
\end{array}
\]
where we write $\mbs{L}_a.\,t$ for $\mbs{L}_a(t)$.
\end{example}

By construction, thus, NEPs represent NESs.
\hide{
Each nominal context $\nomctx{\atmlst{a}{n}}{V}$
determines the coarity $\atmprd{n}$ and the arity $\toTES{V}$;
and
Each nominal term~$\nomctx{\atmlst{a}{n}}{V} \vdash t$ determines the map
\[
\toTES{\nomctx{\atmlst{a}{n}}{V} \vdash t} 
:\atmprd{n} \rightarrow T_\Sigma\toTES{V}
\]
corresponding to 
$\atmabs{\atmlst{a}{n}}{t} \in \sephom{\atmprd{n},T_\Sigma\toTES{V}}$
via the bijection~(\ref{NominalKmapBijection}).
Indeed, the equivariant function $\toTES{\nomctx{\atmlst a n}V\vdash t}$
maps $\atmlstz b\in\atmprd{n}$ to 
$\atmtrans{\atmlstz a}{\atmlstz b}\permact t \in T_\Sigma\toTES{V}$.
}
Indeed, a NEP~$\mathbb T = (\Sigma,E)$ determines the NES 
$
\toTES{\mathbb T} = 
(\ul{\Nom},\monad T_\Sigma,\toTES{E})
$
with the set of equations $\toTES{E}$ given by
\[\setof{
\eqnz
{\toTES{\nomctx{\atmlstz a}V\vdash l}}
{\toTES{\nomctx{\atmlstz a}V\vdash r}}
:\atmprd{n} \rightarrow T_\Sigma\toTES{V}
}
_{(\, \eqn{\nomctx{\atmlstz a}V}{l}{r}\,)\,\in\,E}
\] 
where ${\toTES{\nomctx{\atmlstz a}V\vdash t}}$ is the Kleisli map
corresponding to $\atmabs{\atmlstz{a}}{t}$ via the
bijection~(\ref{NominalKmapBijection}).

\hide{
\begin{remark}
\label{rem:eqvar-for-nomthy}
To have a bijection between NESs and NEPs,
one has to take \mh{\alpha} equivalence classes of nominal terms for the
\mh{\alpha} equivalence relation $\approx_\alpha$ generated by the rule
\[
\big(\terminctx{\nomctx{\atmlst{a}{n}}{V}}{t}\big)
\;\approx_\alpha\;
\big(\terminctx{\nomctx{\atmlst{b}{n}}{V}}
{\atmtrans{\atmlst a n}{\atmlst b n}\permact t}\big)
\, .
\]
However, instead of imposing the equivalence on syntactic terms, we take
this into account when we reason about them by introducing the following
rule:
\[
\myproof{
  \AxiomC{$\eqn{\nomctx{\atmlst{a}{n}}{V}}{t}{t'}$}
  \UnaryInfC{$
    \eqn{\nomctx{\atmlst{b}{n}}{V}}
    {\atmtrans{\atmlst a n}{\atmlst b n}\permact t}
    {\atmtrans{\atmlst a n}{\atmlst b n}\permact t'}
    $}
}
\] 
\end{remark}
}

\subsection{Model theory}

The model theory of a NEP~$\mathbb T=(\Sigma,E)$ is derived from that
of the NES~$\toTES{\mathbb T}$.  This we now spell out in elementary
terms.

A \emph{$\mathbb T$-algebra} is an
\mbox{$F_\Sigma$-algebra}~$(M,\ev:F_\Sigma M\rightarrow M)$ such that for
all nominal equations~$(\nomctx{\atmlstz a}V\vdash l\equiv r)\in E$, 
\[
\llrrbrk{\, \nomctx{\atmlstz{a}}{V} \vdash l \,}_{(M,\ev)}
=
\llrrbrk{\, \nomctx{\atmlstz{a}}{V} \vdash r \,}_{(M,\ev)}
\ : 
\llrrbrk{ \nomctx{\atmlstz a}V }(M) \rightarrow M
\]
where 
%
\[
\textstyle
\llrrbrk{\nomctx{\atmlstz a}V}(M)
\;=\;
  \big( 
  \prod_{x\in\U{V}}\atmhom{V(x)}{M} 
  \big)
  \septen
  \atmprd{n}
\]
and 
where 
$\llrrbrk{\, \nomctx{\atmlstz{a}}{V} \vdash t \,}_{(M,\ev)}
$
is inductively defined as follows:
\begin{enumerate}[$\bullet$]
\item
$\llrrbrk{\,\nomctx {\atmlstz a} V\vdash {x(\atmlstz{b})}\,}_{(M,\ev)}
\big(\setof{\atmabs{\atmlstz{c}_x}m_x}_{x\in\U{V}},\,\atmlstz d\big)
=
\atmtrans{\atmlstz{c}_x}{\atmlstz{c}}\permact m_x$
for $\atmlstz{c} = \atmtrans{\atmlstz a}{\atmlstz d}\permact
\atmlstz{b}$,
\\[-1mm]

\item
$\llrrbrk{\,\nomctx {\atmlstz a} V \vdash {\opr(t_1,\ldots,t_k)}\,}_{(M,\ev)}
\big(\setof{\atmabs{\atmlstz{c}_x}m_x}_{x\in\U{V}},\,\atmlstz d\big)
=
\ev_k(\opr',t'_1,\ldots,t'_k)
$
for ${\ev_k:\Sigma(k)\times M^{k}\rightarrow M}$ the $k$-component of the
structure map $\ev$, and 
\[
\opr' 
= 
\atmtrans{\atmlstz a}{\atmlstz d}\permact\opr
\enspace ,
\quad
t'_i 
= 
\llrrbrk{\,\nomctx{\atmlstz a} V \vdash {t_i}\,}_{(M,\ev)}
\big(\setof{\atmabs{\atmlstz{c}_x}m_x}_{x\in\U{V}},\,\atmlstz d\big)
\enspace.
\]
\end{enumerate}
%
The category $\catalg{\mathbb{T}}$ is the full subcategory of
$\catalg{F_\Sigma}$ consisting of $\mathbb{T}$-algebras.
$\catalg{\mathbb{T}}$ and $\catalg{\toTES{\mathbb T}}$ are isomorphic by
construction. 

\begin{example}[continued from Example~\ref{ex:nom-thy-lam}]
\label{ex:nom-thy-lam-model}
A \mh{\mathbb{T}_\lambda} algebra has a carrier $M\in\Nom$ with structure
maps 
\[
\llrrbrk{\mbs{V}} : \NomAtom \rightarrow M \enspace,\quad
\llrrbrk{\mbs{L}} : \NomAtom\times M\rightarrow M \enspace ,\quad
\llrrbrk{\mbs{A}} : M^2 \rightarrow M
\]
satisfying the equations of the theory.  For instance, according to the
equation~$(\alpha)$, we have that
\[
\llrrbrk{\mbs{L}}\big(a,\atmtrans c a\permact m\big)
= 
\llrrbrk{\mbs{L}}\big(b,\atmtrans c b\permact m\big)
\quad \text{ for all }
(\atmabs c m,(a,b))\in \sephom{\NomAtom,M}\septen \atmprd{2}
\] 
and, according to the equation~$(\eta)$, we have that
\[
\llrrbrk{\mbs{L}}
  \big(a,\llrrbrk{\mbs{A}}(m,\llrrbrk{\mbs{V}}(a))\big)
=
m
\quad \text{ for all } (m,a)\in M\septen\NomAtom
\enspace.
\]

\hide{
By examining the construction~(\ref{eqn:qt-con-simp}) of
the free \mh{\mathbb{T}_\lambda} algebra over 
the initial \mh{F_{\Sigma_\lambda}} algebra
$T_{\Sigma_\lambda}(0)$ with the syntactic structure map,
one can see that 
the initial \mh{\mathbb{T}_\lambda} algebra has
as carrier the nominal set of 
\mh{\alpha\beta\eta} equivalence classes of
\mh{\lambda} terms with the appropriate \mh\Perm action.
}
\end{example}

\subsection{Synthetic Nominal Equational Logic}

For a NEP~${\mathbb{T}= (\Sigma,E)}$, we consider the EML associated to the
NES~$\toTES{\mathbb T}$ in syntactic form, and thereby synthesise a deductive
system for deriving valid nominal equations in \mbox{$\mathbb T$-algebras}.
The resulting \emph{Synthetic Nominal Equational Logic~(SNEL)} has the
inference rules given in Figure~\ref{fig:SNELRules}.  
\begin{figure}[t]
\[\begin{array}{c}
\myproof{
  \AxiomC{$\Eeqn{\nomctx{\atmlst{a}{n}}{V}}{t}{t'}$}
  \LeftLabel{$\s{Eqvar}$}
  \UnaryInfC{$
    \Eeqn{\nomctx{\atmlst{b}{n}}{V}}
    {\atmtrans{\atmlst a n}{\atmlst b n}\permact t}
    {\atmtrans{\atmlst a n}{\atmlst b n}\permact t'}
    $}
}
\\[9mm]
\myproof{
  \AxiomC{$\nomctx{\atmlsth a n}V\vdash_E t$}
  \LeftLabel{$\s{Ref}$}
  \UnaryInfC{$\Eeqn{\nomctx{\atmlsth a n}V}{t}{t}$}
}
\qquad  
\myproof{
\AxiomC{$\Eeqn{\nomctx{\atmlsth a n}V}{t}{t'}$}
\LeftLabel{$\s{Sym}$}
\UnaryInfC{$\Eeqn{\nomctx{\atmlsth a n}V}{t'}{t}$}
}
\qquad 
\myproof{
\AxiomC{${
\Eeqn{\nomctx{\atmlsth a n}V}{t}{t'}
\quad 
\Eeqn{\nomctx{\atmlsth a n}V}{t'}{t''}
}$}
\quad
\LeftLabel{$\s{Trans}$}
\UnaryInfC{$\Eeqn{\nomctx{\atmlsth a n}V}{t}{t''}$}
}
\\[9mm]
\myproof{
\AxiomC{$\big(\eqn{\nomctx{\atmlsth a n} {V}}{l}{r}\big)\in E$}
\LeftLabel{$\s{Axiom}$}
\UnaryInfC{$\Eeqn{\nomctx{\atmlsth a n}V}{l}{r}$}
}
\qquad\qquad
\myproof{
  \AxiomC{$\Eeqn{\nomctx{\atmlstz a,\atmlstz b} V}{t}{t'}$}
  \RightLabel{\big($\atmlstz b  \freshfor \atmlstz{a}, t,t'$\big)}
  \LeftLabel{$\s{Elim}$}
  \UnaryInfC{$\Eeqn{\nomctx {\atmlstz a} V}{t}{t'} $}
}
\\[9mm]
\begin{array}{c}
\myproof{
  \AxiomC{$\Eeqn{\nomctx {\atmlstz a} V}{t}{t'}$}
  \LeftLabel{$\s{Intro}$}
  \UnaryInfC{${\Eeqn
               {\nomctx{\atmlstz a,\atmlst b m}{\Vext V m}}
               {\; t\subst{x(\atmlstz{b}_x) 
	           \mapsto 
	           x(\atmlstz{b}_x,\atmlstz b)}_{x\in\U{V}} \;}
               {\; t'\subst{x(\atmlstz{b}_x) 
	           \mapsto 
	           x(\atmlstz{b}_x,\atmlstz b)}_{x\in\U{V}} }
             }$}
}
\\[7mm]
\text{with }
\atmlstz b \freshfor \atmlstz a
\text{ and }
\forall_{x\in\U{V}}\; 
\atmlstz b\freshfor \atmlstz b_x\ 
\\[1.5mm]
\text{where }
\U{\Vext V m} \,=\, \U{V}
\text{ and }
\forall_{x\in\U{V}}\; \Vext V m (x)
\,=\, V(x) + m
\end{array}
\\[16mm]
\myproof{
\AxiomC{$\Eeqn{\nomctx {\atmlstz a} V}{t}{t'} 
         \qquad
         \Eeqn{\nomctx {{\atmlstz{b}_x}^{V(x)}}\, U }{s_x}{s'_x}
	   \enspace(x\in \U{V})$}
\LeftLabel{$\s{Subst}_\amalg$}
\UnaryInfC{$
\Eeqn
{\nomctx{\atmlstz a}\, U \;}
{\;
  t
 \subst {x(\atmlstz{b}_x) \mapsto s_x}_{x\in\U{V}}
 \;}
{\;
 t'
 \subst {x(\atmlstz{b}_x) \mapsto s'_x}_{x\in\U{V}}
}
$}
}
\end{array}\]
\caption{Rules of SNEL.
}
\label{fig:SNELRules}
\hrulefill
\end{figure}
The \emph{substitution operation} in the rules~$\s{Intro}$ and
$\s{Subst}_\amalg$ maps\\[-2mm]
\[
t \in T_\Sigma \toTES{U}
\enspace,\quad 
\setof{\,
  \atmabs{\atmlstz{c}_x} s_x  \in \atmhom{U(x)}{T_\Sigma X}
\,}_{x\in\U{U}}
\]\\[-3mm]
to the nominal term\\[-2.5mm]
\[
t\subst{x({\atmlstz c}_x) \mapsto s_x}_{x\in\U{U}}
\,\in\,
T_\Sigma X
\]\\[-2mm]
defined by structural induction on $t$ as follows:\\[-2mm]
\[
\begin{array}{rcl}
x(\atmlstz a)
\subst{x({\atmlstz c}_x) \mapsto s_x}_{x\in\U{U}}
& = &
(\atmlstz c_x \ \atmlstz{a})\permact s_x
\\[2mm]
\opr(t_1,\ldots,t_k) 
\subst{x({\atmlstz c}_x) \mapsto s_x}_{x\in\U{U}}
& = &
\opr\big(
t_1
\subst{x({\atmlstz c}_x) \mapsto s_x}_{x\in\U{U}}
,
\ldots, 
t_k
\subst{x({\atmlstz c}_x) \mapsto s_x}_{x\in\U{U}}
\big)
\enspace.
\end{array}
\]\\[-5mm]

\begin{remark}\label{SNELDerivableRules}
Note that under the rule~$\s{Ref}$, the rules~$\s{Intro}$ and
$\s{Subst}_\amalg$ are inter-derivable with the rule\\[.5mm]
\[\myproof{
\LeftLabel{$\s{IntroSubst}_\amalg$}
\AxiomC
{$\Eeqn{\nomctx{\atmlstz a}V}{t}{t'}
   \qquad
   \Eeqn{\nomctx{{\atmlstz b_x}^{V(x)},\atmlstz b}\,U}{s_x}{s'_x}
   \enspace(x\in\U V)$}
\RightLabel{$(\atmlstz b\freshfor\atmlstz a)$}
\UnaryInfC
  {$\Eeqn
      {\nomctx{\atmlstz a,\atmlstz b}\,U}
      {t\subst{x(\atmlstz b_x)\mapsto s_x}_{x\in\U V}}
      {t'\subst{x(\atmlstz b_x)\mapsto s'_x}_{x\in\U V}}
      $}
}\]\\[1mm]
Indeed, the above arises from the rule~$\s{Intro}$ for $\atmlstz b$ on the
judgement $\Eeqn{\nomctx{\atmlstz a}V}{t}{t'}$ followed
by the rule~$\s{Subst}_\amalg$ with respect to the family
$\Eeqn{\nomctx{{\atmlstz b_x}^{V(x)},\atmlstz b}\,U}{s_x}{s'_x}$ ($x\in\U V$);
whilst, conversely, the rule~$\s{Subst}_\amalg$ is the special case of the
rule~$\s{IntroSubst}_\amalg$ for $\atmlstz b$ the empty tuple and the
rule~$\s{Intro}$ arises by instantiating the rule~$\s{IntroSubst}_\amalg$ with
the family 
$\Eeqn
   {\nomctx{{\atmlstz b_x}^{V(x)},\atmlst b m}{\Vext V m}}
   {x(\atmlstz b_x,\atmlstz b)}
   {x(\atmlstz b_x,\atmlstz b)}$ 
($x\in\U V$). 

We also note that the rule~$\s{Elim}$ is in fact reversible, as the
instantiation of the rule $\s{IntroSubst}_\amalg$ with the family 
$\Eeqn
   {\nomctx{{\atmlstz b_x}^{V(x)},\atmlstz b}V}
   {x(\atmlstz b_x)}
   {x(\atmlstz b_x)}$ 
($x\in\U V$)
yields the derivability of the rule\\[-1mm]
\[\myproof{
\LeftLabel{$\s{Inc}$}
\AxiomC{$\Eeqn{\nomctx{\atmlstz a} V} t {t'}$}
\RightLabel{$(\atmlstz b \freshfor \atmlstz a)$}
\UnaryInfC{$\Eeqn{\nomctx{\atmlstz a,\atmlstz b} V} t {t'}$}
}\]\\[1mm]
\end{remark}\mbox{}\\[-10mm]

\begin{example}[continued from Example~\ref{ex:nom-thy-lam-model}]
\label{ex:proof-by-SNEL}
We give a derivation of 
\[
\eqn
{\nomctx{a}{\;x:1,y:0} }
{\mbs{A}(\mbs{L}_a.\,\mbs{L}_a.\,x(a),y) \,}
{\, \mbs{L}_a.\,x(a)}
\]
in the SNEL of $\mathbb T_\lambda$: 
\\[4mm]
\mbox{}\hfill$
\myproof{
\AxiomC{$
  \setlength{\arraycolsep}{2pt}
  \begin{array}{rrcll}
   &
  \eqn
  {\nomctx{a,b}{\;x:1} &}
  {& \mbs{L}_a.\,x(a) \,}
  {\, \mbs{L}_b.\,x(b)}
  &
  \ \text{by }\s{Axiom}~(\alpha)
  \\
  x\mapsto&
  \eqn
  {\nomctx{c}{\;x:1,y:0} &}
  {& x(c) \,}
  {\, x(c)}
  &
  \ \text{by }\s{Ref}
  \end{array}
$}
\LeftLabel{$\mbf{A}:\ $}
\RightLabel{by $\s{Subst}_\amalg$}
\UnaryInfC{$
  \eqn
  {\nomctx{a,b}{\;x:1,y:0} }
  { \mbs{L}_a.\,x(a) \,}
  {\, \mbs{L}_b.\,x(b)}
  $}
}
$\hfill\mbox{}
\\[4mm]
\mbox{}\hfill$
\myproof{
\AxiomC{$
  \setlength{\arraycolsep}{2pt}
  \begin{array}{rrcll}
   &
  \eqn
  {\nomctx{a,b}{\;z:2},w:0 &}
  {& \mbs{A}(\mbs{L}_a.z(a,b),w) \,}
  {\, \mbs{A}(\mbs{L}_a.z(a,b),w)}
  &
  \ \text{by }\s{Ref}
  \\
  z\mapsto&
  \eqn
  {\nomctx{a,b}{\;x:1,y:0} &}
  {& \mbs{L}_a.\,x(a) \,}
  {\, \mbs{L}_b.\,x(b)}
  &
  \ \text{by }\mbf{A}
  \\
  w\mapsto&
  \eqn
  {\nomctx{\;}{\;x:1,y:0} &}
  {& y \,}
  {\, y}
  &
  \ \text{by }\s{Ref}
  \end{array}
$}
\LeftLabel{$\mbf{B}:\ $}
\RightLabel{
by $\s{Subst}_\amalg$
}
\UnaryInfC{$
  \eqn
  {\nomctx{a,b}{\;x:1,y:0} }
  {\mbs{A}(\mbs{L}_a.\,\mbs{L}_a.\,x(a),y) \,}
  {\, \mbs{A}(\mbs{L}_a.\,\mbs{L}_b.\,x(b),y)}
  $}
}
$\hfill\mbox{}
\\[4mm]
\mbox{}\hfill$
\myproof{
  \AxiomC{$
    \eqn
    {\nomctx{a}{\; x:0,y:1 } }
    { \mbs{A}\big(\mbs{L}_a.\,x\,,\,y(a)\big) \,}
    {\, x}
    \quad \text{by }\s{Axiom}~(\beta_\kappa)
    $}
  \LeftLabel{$\mbf{C}:$}
  \RightLabel{by $\s{Intro}$}
  \UnaryInfC{$
    \eqn
    {\nomctx{a,b}{\; x:1,y:2 } }
    {\mbs{A}\big(\mbs{L}_a.\,x(b)\,,\,y(a,b)\big) \,}
    {\, x(b)}
    $}
}
$\hfill\mbox{}
\\[4mm]
\mbox{}\hfill$
\myproof{
\AxiomC{$
  \setlength{\arraycolsep}{2pt}
  \begin{array}{rrcll}
   &
    \eqn
    {\nomctx{a,b}{\; x:1,y:2 } &}
    {& \mbs{A}\big(\mbs{L}_a.\,x(b)\,,\,y(a,b)\big) \,}
    {\, x(b)}
  &
  \ \text{by }\mbf{C}
  \\
  x\mapsto&
  \eqn
  {\nomctx{b}{\;x:1,y:0} &}
  {& \mbs{L}_b.\,x(b) \,}
  {\, \mbs{L}_b.\,x(b)}
  &
  \ \text{by }\s{Ref}
  \\
  y\mapsto&
  \eqn
  {\nomctx{a,b}{\;x:1,y:0} &}
  {& y \,}
  {\, y}
  &
  \ \text{by }\s{Ref}
  \end{array}
$}
\LeftLabel{$\mbf{D}:\ $}
\RightLabel{
by $\s{Subst}_\amalg$
}
\UnaryInfC{$
  \eqn
  {\nomctx{a,b}{\;x:1,y:0} }
  {\mbs{A}(\mbs{L}_a.\,\mbs{L}_b.\,x(b),y) \,}
  {\, \mbs{L}_b.\,x(b)}
  $}
}
$\hfill\mbox{}
\\[4mm]
\mbox{}\hfill$
\myproof{
\AxiomC{$
  \!
  \eqn
  {\nomctx{a,b}{\;x:1,y:0} }
  {\, \mbs{A}(\mbs{L}_a.\,\mbs{L}_a.\,x(a),y)}
  {\, \mbs{L}_a.\,x(a)}
  \ \ 
  \text{by }\s{Trans}(\s{Trans}(\mbf{B},\mbf{D}),\s{Sym}(\mbf{A}))
  \!
$}
\RightLabel{by $\s{Elim}$}
\UnaryInfC{$
  \eqn
  {\nomctx{a}{\;x:1,y:0} }
  {\, \mbs{A}(\mbs{L}_a.\,\mbs{L}_a.\,x(a),y)}
  {\, \mbs{L}_a.\,x(a)}
  $}
}
$\hfill\mbox{}
\end{example}

\subsection{Soundness}

By construction, if a nominal equation 
$\Eeqn {\nomctx{\atmlsth a n}V} t {t'}$ is derivable in SNEL,
then the equation $\eqn{\toTES{E}}
{\toTES{\terminctx{\nomctx{\atmlsth a n}V}{t}}}
{\toTES{\terminctx{\nomctx{\atmlsth a n}V}{t'}}}$ 
is derivable in EML.
We explain why this is so for each rule.
\begin{enumerate}[$\bullet$]
\item
  The SNEL rule~$\s{Eqvar}$ arises from the fact that 
  \[\ul{\nomctx{\atmlst a n}V\vdash s}
   =
   \ul{\nomctx{\atmlst b n}V\vdash
     \atmtrans{\atmlst a n}{\atmlst b n}\permact s}
     : {\Atom^{\# n}\rightarrow T_\Sigma(\ul V)}\] 
  for all $\atmlst b n \in\atmprd{n}$ and nominal terms~$\nomctx{\atmlst a
  n}V\vdash s$.\footnote{The omission of this rule in the SNEL presented
  in~\citet{FioreHur08} is an oversight.}
  
\medskip
\item
The SNEL~rules $\s{Ref}$, $\s{Sym}$, $\s{Trans}$, and $\s{Axiom}$ directly
mimic the corresponding EML rules.

\medskip
\item
The SNEL rule~$\s{Elim}$ arises from the EML rule~$\s{Local}_1$ with respect
to the epimorphic projection map 
$\atmprd{(n+m)}
 \xymatrix@C=15pt{\ar@{->>}[r]&}
 \atmprd{n}$
sending $(\atmlst a n,\atmlst b m)$ to $(\atmlst a n)$.

\medskip
\item
The SNEL rule~$\s{Intro}$ arises from the EML rule~$\s{Ext}$
extended with the nominal set $\atmprd{m}$.

\noindent
Note that for $\nomctx{\atmlstz a}V\vdash s$, one has that
$
\toTES{
\terminctx
{\nomctx{\atmlst a n,\atmlst b m}{\Vext V m}}
{
s
\subst{x(\atmlstz{c}_x) \mapsto x(\atmlstz{c}_x,\atmlstz b)}_{x\in\U{V}}
}
}
$
amounts to the composite\\[2mm]
$
\xymatrix@C=6pc{
\atmprd{(n+m)}
\iso
\atmprd{m}\septen\atmprd{n}
\ar[r]^-{
\langle \atmprd{m} \rangle \,
\toTES{\nomctx{\atmlstz a}{V}\vdash s}
}
&
T_\Sigma\big(\atmprd{m}\septen\toTES{V} \big)
\iso
T_\Sigma\big(\coprod_{x\in\U V}\atmprd{(V(x)+m)}\big)
}
.
$

\medskip
\item
The SNEL rule~$\s{Subst}_\amalg$ arises from the EML rule\\
\[\quad\qquad
\begin{array}{c}
\myproof{
  \AxiomC{$
    \eqnz
    {\toTES{\terminctx{\nomctx{\atmlstz a}{U}}{t}}}
    {\toTES{\terminctx{\nomctx{\atmlstz a}{U}}{t'}}}
  \qquad
    \eqnz
    {\toTES{\terminctx{\nomctx{\atmlstz{b}_x}{V}}{s_x}}\,}
    {\,\toTES{\terminctx{\nomctx{\atmlstz{b}_x}{V}}{s'_x}}}
    \enspace(x\in\U{U})
    $}
  \LeftLabel{$\s{Subst}_\amalg$}  
  \UnaryInfC{$
    \eqnz
    {\big(\toTES{\terminctx{\nomctx{\atmlstz a}{U}}{t}}\big)
    \Big\{
    \big[
    \toTES{\terminctx{\nomctx{\atmlstz{b}_x}{V}}{s_x}}
    \big]_{x\in\U{U}}
    \Big\}
    \,}
    {\,\big(\toTES{\terminctx{\nomctx{\atmlstz a}{U}}{t'}}\big)
    \Big\{
    \big[
    \toTES{\terminctx{\nomctx{\atmlstz{b}_x}{V}}{s'_x}}
    \big]_{x\in\U{U}}
    \Big\}
    }
    $}
}
\end{array}\]\\[2mm]
noting that 
$
\big(\toTES{\terminctx{\nomctx{\atmlstz a}{U}}{t}}\big)
\Big\{
\big[
\toTES{\terminctx{\nomctx{\atmlstz{b}_x}{V}}{s_x}}
\big]_{x\in\U{U}}
\Big\}
\;=\;
\toTES{
\terminctx
{\nomctx{\atmlstz a} V }
{
  t
 \subst {x(\atmlstz{b}_x) \mapsto s_x}_{x\in\U{U}}
}
}
$.
\end{enumerate}
\bigskip
Thus, the soundness of SNEL follows from that of EML.

\hide{
\begin{remark}
Since the category of sets embeds in that of nominal sets, every algebraic
theory is a NEP and for them SNEL restricted to contexts with empty atom
context and variables of valence zero reduces to the equational logic SEL
for algebraic theories.  
\end{remark}
}

\subsection{Completeness}
\label{NominalCompleteness}

We provide a sound and complete rewriting-style deduction system for NEPs,
referred to as \emph{Synthetic Nominal Rewriting~(SNR)}, and establish the
completeness of SNEL.

For every NEP~$\mathbb{T}$, the associated NES~$\toTES{\mathbb T}$ is
inductive.  Indeed, using that in $\Nom$ finite limits commute with
filtered colimits and equivariant functions are epimorphic iff their
underlying function is surjective, one sees that the
endofunctor~$F_\Sigma$ associated to a nominal signature~$\Sigma$ 
preserves colimits of \mh{\omega}chains and epimorphisms.
Moreover, since for every~$n\in\Nat$, the right 
adjoint~$\inthom{\atmprd{n},-}$ is also a left adjoint, it follows that,
for every variable context~$V$, the nominal
set~$\toTES{V}$ is 
compact 
and projective.

For a NEP~${\mathbb{T}=(\Sigma,E)}$, we consider the
construction~(\ref{eqn:qt-con-simp}) for the associated
NES~$\toTES{\mathbb T}$.  Since the forgetful
functor~${\U{-}:\Nom\rightarrow\Set}$ creates colimits, we have the
following explicit description. 
For a nominal set~$X$, the nominal set~$(T_\Sigma X)_1$ has as underlying
set the quotient~$\U{T_\Sigma X}/_{\approx_1}$ under the equivalence
relation~$\approx_1$ on $\U{T_\Sigma X}$ generated by the following
rule:\\
\[
\myproof{%
\AxiomC{$\big(\eqn{\nomctx {\atmlst a n} {V}}{l}{r}\big)\in E$}
\UnaryInfC{$\big(\atmtrans{\atmlst a n}{\atmlst b n}\permact l\big)  
            \subst{x(\atmlstz c_x) \mapsto s_x}_{x\in \U V}
            \approx_1
            \big(\atmtrans{\atmlst a n}{\atmlst b n}\permact r\big)  
            \subst{x(\atmlstz c_x) \mapsto s_x}_{x\in \U V}$}
}
\]\\[1mm]
for $\atmlsth b n \in\atmprd{n}$ and 
$\atmabs{\atmlstz{c}_x}s_x\in \atmhom{V(x)}{T_\Sigma X}$
such that $\atmlsth b n \freshfor \atmabs{\atmlstz{c}_x}s_x$
for all $x\in\U{V}$.
The action of $(T_\Sigma X)_1$ is given by 
${\pi\permact [t]_{\approx_1} = [\pi\permact t]_{\approx_1}}$.
The equivariant function $q_0:T_\Sigma X\epirightarrow (T_\Sigma X)_1$
maps a nominal term~${t}$ to its equivalence class~$[t]_{\approx_1}$.

The nominal set $(T_\Sigma X)_n$, for $n \ge 2$, has as underlying set the
quotient $\U{T_\Sigma X}/_{\approx_n}$ under the equivalence
relation~$\approx_n$ on $\U{T_\Sigma X}$ generated by the following
rules:\\
\[
\myproof{%
\AxiomC{$s \approx_{n-1} s'$}
\UnaryInfC{$s \approx_{n} s'$}
}
\qquad\qquad
\myproof{%
\AxiomC{$
    s_1 \approx_{n-1} s'_1,\ \ldots,\  
    s_{k} \approx_{n-1} s'_{k}
$}
\RightLabel{\big($\opr\in\Sigma(k)$\big)}
\UnaryInfC{$
  \opr(s_1,\ldots,s_k) 
  \approx_{n}
  \opr(s'_1,\ldots,s'_k)
$}
}
\]\\[1mm]
The action of $(T_\Sigma X)_n$ is given by 
${\pi\permact [t]_{\approx_n} = [\pi\permact t]_{\approx_n}}$.
The equivariant function
$q_{n-1}:(T_\Sigma X)_{n-1} \epirightarrow (T_\Sigma X)_n$
maps ${[t]_{\approx_{n-1}}}$ to $[t]_{\approx_n}$.

The nominal set $T_{\toTES{\mathbb T}} X$, being the colimit of 
the 
chain~${\setof{q_n:(T_\Sigma X)_n\epirightarrow(T_\Sigma X)_{n+1}}_{n\ge 0}}$, 
is given by $\U{T_{\toTES{\mathbb T}} X} = \U{T_\Sigma X}/_{\approx_E}$
with action $\pi\permact [t]_{\approx_E} = [\pi\permact t]_{\approx_E}$,
for $\approx_E$ the equivalence relation on $\U{T_\Sigma X}$ specified 
by the rules of Figure~\ref{fig:rw-logic-for-nomthy}.  The quotient 
map~${\qt^{\toTES{\mathbb{T}}}_X:T_\Sigma X\rightarrow T_{\toTES{\mathbb T}}}X$
sends a nominal term~$t$ to its equivalence class~$[t]_{\approx_E}$.
\begin{figure}[t]
\noindent\mbox{}
$\displaystyle
\begin{array}{c}
\myproof{
  \AxiomC{$t \in \U{T_\Sigma X}$}
  \LeftLabel{$\s{Ref}$}
  \UnaryInfC{$t \approx_{E} t$}
}
\qquad
\myproof{
  \AxiomC{$t \approx_{E} t'$}
  \LeftLabel{$\s{Sym}$}
  \UnaryInfC{$t' \approx_{E} t$}
}
\qquad
\myproof{
  \AxiomC{$t \approx_{E} t' \quad t' \approx_{E} t''$}
  \LeftLabel{$\s{Trans}$}
  \UnaryInfC{$t \approx_{E} t''$}
}
\\[9mm]
\begin{array}{c}
\myproof{%
\LeftLabel{$\s{Inst}$}
\AxiomC{$\big(\eqn{\nomctx {\atmlst a n} {V}}{l}{r}\big)\in E$}
\UnaryInfC{$
  \big(\atmtrans{\atmlst a n}{\atmlst b n}\permact l\big)  
  \subst{x(\atmlstz c_x) \mapsto s_x}_{x\in \U V}
  \approx_E
  \big(\atmtrans{\atmlst a n}{\atmlst b n}\permact r\big)  
  \subst{x(\atmlstz c_x) \mapsto s_x}_{x\in \U V}$}
}
\\[7mm]
\begin{array}{c}
\text{with }
\atmlsth b n \in\atmprd{n}
\text{ and }
\atmabs{\atmlstz{c}_x}s_x\in \atmhom{V(x)}{T_\Sigma X}
\text{ such that }
\atmlsth b n \freshfor \atmabs{\atmlstz{c}_x}s_x
\text{ for all }
x\in\U{V}
\end{array}
\end{array}
\\[13.5mm]
\myproof{%
\AxiomC{$
    s_1 \approx_E s'_1,\ \ldots,\  
    s_{k} \approx_E s'_{k}
$}
\LeftLabel{$\s{Cong}$}
\RightLabel{\big($\opr\in\Sigma(k)$\big)}
\UnaryInfC{$
  \opr(s_1,\ldots,s_k) 
  \approx_E
  \opr(s'_1,\ldots,s'_k)
$}
}
\end{array}
$\hfill\\
\caption{Rules for $\approx_E$.}
\label{fig:rw-logic-for-nomthy}
\hrulefill
\end{figure}

The rules~$\s{Inst}$ and $\s{Cong}$ of Figure~\ref{fig:rw-logic-for-nomthy}
can be merged into a single one to yield a rewriting-style deduction system.
Indeed, by an induction on the depth of proof trees, one shows that the
relation~$\approx_E$ coincides with the equivalence 
relation~${\approx^{\s R}_E} \subseteq (T_\Sigma X)^2$ generated by the
rule\\[-1mm]
\begin{equation}
\label{eqn:rw-rule-for-nomthy}
\mbf{C}\big[
\big(\atmtrans{\atmlst a n}{\atmlst b n}\permact l\big)  
\subst{x(\atmlstz c_x) \mapsto s_x}_{x\in \U V}
\big]
\;\approx^\s{R}_E\;
\mbf{C}\big[
\big(\atmtrans{\atmlst a n}{\atmlst b n}\permact r\big)  
\subst{x(\atmlstz c_x) \mapsto s_x}_{x\in \U V}
\big]
\end{equation}\\[-1mm]
for $\big(\eqn{\nomctx {\atmlst a n} {V}}{l}{r}\big)\in E$ and
$\mbf{C}[-]$ a context with one hole and possibly with elements from $X$,
and for $\atmlsth b n \in\atmprd{n}$ and 
$\atmabs{\atmlstz{c}_x}s_x\in \atmhom{V(x)}{T_\Sigma X}$ such that
$\atmlst b n \freshfor \atmabs{\atmlstz{c}_x}s_x$ for all $x\in\U{V}$.
Henceforth, the rewriting of nominal terms by the
rule~(\ref{eqn:rw-rule-for-nomthy}) is referred to as \emph{Synthetic
Nominal Rewriting~(SNR)}.

\begin{example}[\cf~the derivation given in
  Example~\ref{ex:proof-by-SNEL}] 
We give a derivation of
\[
\mbs{A}(\mbs{L}_a.\,\mbs{L}_a.\,x(a),y)
\approx^\s{R}
\mbs{L}_a.\,x(a)
\text{ in }
T_{\Sigma_\lambda}(\toTES{x:1,y:0})
\]
%
%
in the SNR of $\mathbb T_\lambda$:\\
\[
\begin{array}[b]{cll}
&
\hspace{-10mm}
\mbs{A}(\mbs{L}_a.\,\mbs{L}_a.\,x(a),y)
\\
&
&\text{by }(\alpha):
\begin{array}{crrl}
&\mbs{A}(\mbs{L}_a.\big[& \mbs{L}_a.\,x(a)\subst{x(c)\mapsto x(c)}&\big],y)
\\
\approx^\s{R}
\!\!\!&&&
\\
&\mbs{A}(\mbs{L}_a.\big[& \mbs{L}_b.\,x(b)\subst{x(c)\mapsto x(c)}&\big],y)
\end{array}
\\
\approx^\s{R}
&
\mbs{A}(\mbs{L}_a.\,\mbs{L}_b.\,x(b),y)
\\
&
&\text{by }(\beta_\kappa):
\begin{array}{crrl}
&\big[& \mbs{A}(\mbs{L}_a.\,x,y(a))\subst{x \mapsto \mbs{L}_b.\,x(b);y(a)\mapsto y} &\big]
\\
\approx^\s{R}
\!\!\!&&&
\\
&\big[& x\subst{x \mapsto \mbs{L}_b.\,x(b);y(a)\mapsto y} &\big]
\end{array}
\\
\approx^\s{R}
&
\mbs{L}_b.\,x(b)
\\
&
&\text{by }(\alpha):
\begin{array}{crrl}
&\big[&\mbs{L}_b.\,x(b)\subst{x(c)\mapsto x(c)}&\big]
\\
\approx^\s{R}
\!\!\!&&&
\\
&\big[&\mbs{L}_a.\,x(a)\subst{x(c)\mapsto x(c)}&\big]
\end{array}
\\
\approx^\s{R}
&
\mbs{L}_a.\,x(a)
\enspace .
\end{array}
\]
\end{example}

The soundness and completeness of SNR is established by means of the
internal completeness of the NES~$\toTES{\mathbb T}$: 
\begin{equation*}
\label{eqn:comp-rw-for-nomthy}
\begin{array}{cl}
&
\hspace{-20mm}
\eqn{\nomctx{\atmlst a n}{V}}{s}{t}
\text{ is satisfied by all \mh{\mathbb T} algebras}
\\[2mm]
\Longleftrightarrow&
\catalg{\toTES{\mathbb{T}}} \;\models\;
\eqnz
{\toTES{\terminctx{\nomctx{\atmlst a n}{V}}{s}}}
{\toTES{\terminctx{\nomctx{\atmlst a n}{V}}{t}}}
\;:\; \atmprd{n} \rightarrow T_\Sigma \toTES{V}
\\[2mm]
\Longleftrightarrow&
\qt^{\toTES{\mathbb T}}_{\toTES{V}} 
\comp \toTES{\terminctx{\nomctx{\atmlst a n}{V}}{s}}
\;=\;
\qt^{\toTES{\mathbb T}}_{\toTES{V}} 
\comp \toTES{\terminctx{\nomctx{\atmlst a n}{V}}{t}}
\;:\;
\atmprd{n} \rightarrow T_{\toTES{\mathbb T}} \toTES{V}
\\[2mm]
\Longleftrightarrow&
\qt^{\toTES{\mathbb T}}_{\toTES{V}} 
\big(\toTES{\terminctx{\nomctx{\atmlst a n}{V}}{s}} (\atmlst a n)\big)
\;=\;
\qt^{\toTES{\mathbb T}}_{\toTES{V}} 
\big(\toTES{\terminctx{\nomctx{\atmlst a n}{V}}{t}} (\atmlst a n)\big)
\quad\text{in } 
\U{T_{\toTES{\mathbb T}} \toTES{V}}/_{\approx_E}
\\[2mm]
\Longleftrightarrow&
[s]_{\approx_{E}} = [t]_{\approx_{E}} 
\quad\text{in } 
\U{T_{\toTES{\mathbb T}} \toTES{V}}/_{\approx_E}
\\[2mm]
\Longleftrightarrow&
s \approx_{E} t 
\quad\text{in } 
\U{T_{\toTES{\mathbb T}} \toTES{V}}
\\[2mm]
\Longleftrightarrow&
s \approx^\s{R}_{E} t 
\quad\text{in } 
\U{T_{\toTES{\mathbb T}} \toTES{V}}
\;.
\end{array}
\end{equation*}\\[-5mm]

The completeness of SNEL follows, as for all ${t,t'\in (T_\Sigma \toTES
U)_{\atmlstz d}}$ every proof of $t\approx_E t'$ can be turned into a
proof of ${\Eeqn {\nomctx{\atmlstz d}U} t {t'}}$ in SNEL.  
In particular, concerning the rule~$\s{Inst}$, for every
$(\eqn{\nomctx{\atmlst a n}V}{l}{r})\in E$, $\atmlsth b n\in\atmprd n$, and
$s_x\in(T_\Sigma\toTES U)_{(\atmlst {c_x}{V(x)},\atmlstz c)}$~($x\in\U V$)
with $\atmlstz b \freshfor \atmlstz c$, one deduces
\[
\Eeqn {\nomctx{\atmlstz b,\atmlstz c}U} {t_l} {t_r}\quad,\enspace
\text{for }
t_u 
= \big(\atmtrans{\atmlsth a n}{\atmlsth b n}\permact u\big) 
    \subst{x(\atmlstz c_x)\mapsto s_x}_{x\in\U V}
\enspace,
\]\\[-2.5mm]
by means of the rules~$\s{Axiom}$, $\s{Eqvar}$, $\s{Ref}$,
$\s{IntroSubst}_\amalg$; and subsequently derives 
${\Eeqn{\nomctx{\atmlstz d}U}{t_l}{t_r}}$ by means of the rules~$\s{Elim}$
and/or $\s{Inc}$ for any $\atmlstz d$ such that $t_l,t_r\in (T_\Sigma\toTES
U)_{\atmlstz d}$.


\subsection{Related work}

Algebraic structure and rewriting in a nominal setting have already been
considered in the literature.  \citet{GabbayMathijssen06,GabbayMathijssen07}
and~\citet{CloustonPitts07} introduced an essentially equivalent notion of
nominal algebra and provided sound and complete equational logics for them,
whilst~\citet{FGM04} introduced nominal rewriting.  

Our SNEL and the \emph{Nominal Equational Logic~(NEL)}
of~\citet{CloustonPitts07} are equivalent.  Indeed,
see~\citet[Section~8.2.6]{HurThesis} for a translation between the equality
judgements of SNEL and NEL that respects the corresponding satisfaction
relations.  Thus, by virtue of the associated completeness theorems, SNEL and
NEL establish the same theorems under different syntactic 
formalisms. 

The \emph{Nominal Rewriting~(NR)} of~\citep{
FG07} appears to be a term-rewriting version of NEL.
However, it has the shortcoming of not being complete for nominal equational
reasoning~(see~\cite[Section~8.2.7]{HurThesis}).


Our 
approach allows us to also put the \emph{Equational Logic for Binding
Terms~(ELBT)} of~\citet{Hamana03} in the nominal context.
Whereas SNEL arises from an EML on $\Nom$, ELBT arises from a related EML on
the super-category~$\SetToI$, for $\ISet$ the category of finite sets and
injections.  Crucially, however, the
embedding~$\Nom\,\rightembedding\SetToI$ does not
preserve the epimorphic projection 
maps~$p_{n,m}:\NomAtom^{\#(n+m)}\rightarrow \NomAtom^{\# n}$~($n\geq0,m>0$).
Thus, the only essential difference between SNEL and ELBT is that the latter
lacks the rule~$\s{Elim}$ (which 
arises from the EML rule~$\s{Local}_1$ with respect to~$p_{n,m}$) as it is
unsound.

\section*{\ul{\large Conclusion}}

We have introduced a categorical framework for the synthesis of equational
logics.  
This 
comprises a general abstract notion of equational presentation 
together with an equational deduction system 
that is sound for a canonical model theory.  
In this context, we have also introduced a mathematical methodology for
establishing completeness.  
This 
is based on an internal strong completeness result 
that typically leads, through an analysis of the construction of free
algebras, to a characterisation of satisfiability via a \mbox{rewriting-style}
deduction system embedded within the equational deduction system.

Two applications of our theory and methodology were presented.  They
respectively provide a
rational reconstruction of Birkhoff's Equational Logic and a novel nominal
logic for reasoning about algebraic structure with name-binding operators.  A
further major application was given in~\citet{FioreHurSOEqLog} with the
synthesis of \emph{Second-Order Equational Logic}: a deductive system for
equational reasoning about 
languages with variable binding and parameterised metavariables~(see
also~\citet{FioreMahmoud}). 

The extension of the theory of this paper from the \mbox{mono-sorted} to
the \mbox{multi-sorted} setting requires a more involved categorical
theory~(see~\citet{Fiore08,FioreHur08,HurThesis}).  A yet more
comprehensive extension for a theory of rewriting modulo equations has
also been developed.  

\section*{Acknowledgement}

We are most grateful to Pierre-Louis Curien for the invitation to
contribute to this volume, for his feedback on the first draft of the
paper that 
helped to improve the presentation, and for his patience in waiting for
this revised version.


\appendix

\section{Proof of Theorem~\ref{EMLsoundness}}
\label{EMLsoundnessproof}

\begin{notation}
  For $f:V\otimes A\rightarrow B$, we write $\ol f:A\rightarrow [V,B]$ for the
  transpose of $f$ with respect to the adjunction~$V\otimes-\dashv [V,-]$ and
  $\widehat f:V\rightarrow [A,B]$ for the transpose of $f$ with respect to the
  adjuction~$-\otimes A\dashv [A,-]$.
\end{notation}

\begin{proof}[Proof of Theorem~\ref{EMLsoundness}]
We establish the soundness of each EML rule;~\ie,~that every \mh{\mathbb S}
algebra satisfying the premises of a rule also satisfies the conclusion.
The soundness of the rules $\s{Ref}$, $\s{Sym}$, $\s{Trans}$, and $\s{Axiom}$
is trivial.
To show the soundness of the rule $\s{Subst}$, one uses that, for all
$u_1: C\rightarrow TB$ and $u_2: B\rightarrow TA$, 
\[
\llrrbrk{u_1 \subst{u_2}}_{(X,s)}
=
\llrrbrk{u_1}_{(X,s)} \comp \big(\ol{\llrrbrk{u_2}_{(X,s)}} \tensor C\big)
: [A,X]\tensor C\rightarrow X
\, .
\]
To show the soundness of the rule $\s{Ext}$, one uses that, for all 
$u: C\rightarrow TA$,
\[\llrrbrk{\tensorext{V}{u}}_{(X,s)}
=
\llrrbrk{u}_{(X,s)} \comp (\widehat{\mathbf{p}}\tensor C) \comp 
\alphaact^{-1}_{[V\tensor A,X],V,C}
: [V\tensor A,X]\tensor (V\tensor C)\rightarrow X
\, ,\] 
where 
\[
\mathbf{p}
= 
\big(
\xymatrix@C=5pc{
  \big([V\tensor A,X]\tensor V\big)\tensor A 
  \ar[r]^-{\alphaact_{[V\tensor A,X],V,A}}
  &
  [V\tensor A,X]\tensor (V\tensor A)
  \ar[r]^-{\epsilon^{V\tensor A}_{X}}
  &
  X
}\big)\,
.
\]
Finally, to show the soundness of the rule $\s{Local}$, one uses that, 
for all $u:C\rightarrow TA$ and $e:C'\rightarrow C$, 
\[
\algstr{\llrrbrk{u\comp e}}{(X,s)}
=
\algstr{\llrrbrk{u}}{(X,s)} \comp ([A,X]\tensor e)
:  [A,X]\tensor C'\rightarrow X
\]\\[-1.5mm]
and that, for every jointly epimorphic family
$\setof{e_i:C_i\rightarrow C}_{i\in I}$, the family 
$\setof{ [A,X]\tensor e_i 
         : [A,X]\tensor C_i\rightarrow [A,X]\tensor C}_{i\in I}$ 
is also jointly epimorphic.  
\end{proof}

\section{Proof of Theorem~\ref{thm:int-comp}}
\label{IntCompApp}

We introduce several lemmas before proceeding to prove the theorem.

\begin{notation}
  For $f:V\otimes A\rightarrow B$, we write $\ol f:A\rightarrow [V,B]$ for the
  transpose of $f$ with respect to the adjunction~$V\otimes-\dashv [V,-]$ and
  $\widehat f:V\rightarrow [A,B]$ for the transpose of $f$ with respect to the
  adjuction~$-\otimes A\dashv [A,-]$.
\end{notation}

\begin{lemma}
Let ${\mathbb S =(\ul{\cat C},\monad T,\Ax)}$ be a MES.  For every 
\mh{\mathbb S}{algebra}~$(X,s: T X \rightarrow X)$, the 
\mh{T} algebra~$\big(\inthom{V,X}, s^V:T\inthom{V,X}\rightarrow \inthom{V,X}\big)$,
where $s^V$ is the transpose of
\[
\ul{s^V} = \big(
\xymatrix@C=4pc{
  V\tensor T\inthom{V,X}
  \ar[r]^-{\st_{V,\inthom{V,X}}}
&
  T(V\tensor \inthom{V,X})
  \ar[r]^-{T(\epsilon^V_X)}
&
  TX
  \ar[r]^-{s}
&
  X
}\big)\, ,
\]
is an \mh{\mathbb S}{algebra}.
\end{lemma}
\proof 
That $(\inthom{V,X},s^V)$ is an Eilenberg-Moore algebra follows from
transposing the following identities:
\begin{enumerate}[(1)]
  \item $ \ul{s^V} \comp (V\tensor \eta_{\inthom{V,X}})
    \ =\  \epsilon^V_X \ :\  V \tensor \inthom{V,X} \rightarrow X $,
    \medskip

  \item $ \ul{s^V} \comp (V\tensor\mu_{\inthom{V,X}})
    \ =\  \ul{s^V} \comp (V\tensor T (\ul{s^V})) 
    \ :\  V \tensor T T \inthom{V,X} \rightarrow X $.
\end{enumerate}
To show that every equation in $\Ax$ is satisfied in $(\inthom{V,X},s^V)$,
one uses that, for all ${w:C\rightarrow TA}$,
\[
\algstr{\llrrbrk{w}}{(\inthom{V,X}, s^V)}
\ =\ 
\ol{
\algstr{\llrrbrk{w}}{(X,s)} \comp (\widehat{\mathbf{p}}\tensor C) \comp
\alpha^{-1}_{V,\inthom{A,\inthom{V,X}},C} 
}
\ :\ 
\inthom{A,\inthom{V,X}}\tensor C \rightarrow [V,X]
\]
where
$\mathbf{p}$ denotes the composite
\[
\xymatrix@C=5pc{
  (V\tensor\inthom{A,\inthom{V,X}})\tensor A
  \ar[r]^-{\alpha_{V,\inthom{A,\inthom{V,X}},A}}
  &
}
\!\!\!
\xymatrix@C=3pc{
  V\tensor(\inthom{A,\inthom{V,X}}\tensor A)
  \ar[r]^-{V\tensor \epsilon^A_{\inthom{V,X}}}
  &
}
\!\!\!
\xymatrix@C=2pc{
  V\tensor\inthom{V,X}
  \ar[r]^-{\epsilon^V_X}
  &
  X
}
\, .
\eqno{\qEd}
\]

\begin{lemma}
  For $\,\mathbb S = (\ul{\cat C},\monad T,\Ax)$ a MES admitting free
  algebras, the free \mh{\mathbb S} algebra monad~$\monad T_\mathbb{S}$ on
  $\cat C$ is strong.
The components of the strength $\st^\mathbb{S}$ are given by the unique maps
such that the following diagram commutes:
\[
\xymatrix@C=45pt{ 
V\tensor T T_{\mathbb S} X \ar[d]_-{V\tensor\tau^{\mathbb S}_X} \ar[r]^-{\st_{V,T_{\mathbb S}X}} 
& 
T(V\tensor T_{\mathbb S}X)\ar[r]^-{T(\st^{\mathbb S}_{V,X})}
& 
T T_{\mathbb S}(V\tensor X) \ar[d]^-{\tau^{\mathbb S}_{V\tensor X}}
\\
V\tensor T_{\mathbb S}X\ar@{-->}[rr]^-{\exists!\,\st^{\mathbb S}_{V,X}} & &  T_{\mathbb S}(V\tensor X)
\\
\ar[u]^-{V\tensor\eta^{\mathbb S}_X} V\tensor X \ar[rru]_-{\eta^{\mathbb S}_{V\tensor X}} & 
}
\]
\end{lemma}
\proof 
First note that 
$\ol{\st^\mathbb{S}_{V,X}}:T_{\mathbb S}X\rightarrow \inthom{V,T_{\mathbb S}(V\tensor X)}$
is the unique homomorphic extension of 
$\ol{\eta^{\mathbb
S}_{V\tensor X}}: X \rightarrow \inthom{V,T_\mathbb{S}(V\tensor X)}$
with respect to the \mh{\mathbb S} algebra 
$\big(\inthom{V,T_\mathbb{S}(V\tensor X)}, {(\tau^\mathbb{S}_{V\tensor X})}^V\big)$.

The naturality of $\st^\mathbb{S}$ follows from the fact that, for all
$f:V\rightarrow V'$ and $g:C\rightarrow C'$, the maps
\[
\ol{T_\mathbb{S}(f\tensor g) \comp \st^\mathbb{S}_{V,X}}\ ,\ 
\ol{\st^\mathbb{S}_{V',X'} \comp (f \tensor T_\mathbb{S}(g))}
\ :\ T_\mathbb{S} X \rightarrow \inthom{V,T_\mathbb{S}(V'\tensor X')}
\]
are both an homomorphic extension of 
$\ol{
\eta^{\mathbb S}_{V'\tensor X'}\comp (f\tensor g)
}: X \rightarrow \inthom{V,T_\mathbb{S}(V'\tensor X')}$
with respect to the \mh{\mathbb S} algebra 
$\big(\inthom{V,T_\mathbb{S}(V'\tensor X')},{(\tau^\mathbb{S}_{V'\tensor X'})}^V\big)$.

Three of the four coherence conditions for strength follow from the fact that
the maps
\begin{eqnarray*}
\ol{T_\mathbb{S} (\lambda_X)\comp \st^\mathbb{S}_{I,X}}\ ,\ 
\ol{\lambda_{T_\mathbb{S} X}}
&:& 
T_\mathbb{S} X \rightarrow \inthom{I,T_\mathbb{S} X}
\\
\ol{T_\mathbb{S}(\alpha_{U,V,X}) \comp \st^\mathbb{S}_{U\monten V, X}}\ ,\ 
\ol{\st^\mathbb{S}_{U,V\tensor X} \comp (U\tensor \st^\mathbb{S}_{V,X}) 
\comp \alpha_{U,V,T_\mathbb{S}X}}
&:& 
T_\mathbb{S} X \rightarrow 
\inthom{U\tensor V, T_\mathbb{S} (U\tensor(V\tensor X))}
\\
\ol{\st^\mathbb{S}_{V,X} \comp (V\tensor \mu^\mathbb{S}_X)}\ ,\ 
\ol{\mu^\mathbb{S}_{V\tensor X} \comp T_\mathbb{S}(\st^\mathbb{S}_{V,X}) \comp \st^\mathbb{S}_{V,T_\mathbb{S}X}}
&:& 
T_\mathbb{S}T_\mathbb{S} X \rightarrow \inthom{V,T_\mathbb{S} (V\tensor X)}
\end{eqnarray*}
are respectively homomorphic extensions of
\begin{eqnarray*}
\ol{\eta^\mathbb{S}_X \comp \lambda_X} &:& 
 X \rightarrow \inthom{I,T_\mathbb{S} X}
\\
\ol{\eta^\mathbb{S}_{U\tensor(V\tensor X)} \comp \alpha_{U,V,X}} &:& 
 X \rightarrow \inthom{U\tensor V, T_\mathbb{S} (U\tensor(V\tensor X))}
\\
\ol{\st^\mathbb{S}_X} &:& 
T_\mathbb{S} X \rightarrow \inthom{V,T_\mathbb{S} (V\tensor X)}
\end{eqnarray*}
with respect to the \mh{\mathbb S} algebras 
$\big(\inthom{I,T_\mathbb{S} X},{(\tau^\mathbb{S}_X)}^I\big)$,
$\big(\inthom{U\tensor V, T_\mathbb{S} (U\tensor(V\tensor X))},
{(\tau^\mathbb{S}_{U\tensor(V\tensor X)})}^{(U\tensor V)}\big)$,
$\big(\inthom{V,T_\mathbb{S} (V\tensor X)},{(\tau^\mathbb{S}_{V\tensor X})}^V\big)$.
The remaining coherence condition is the triangle in the diagram above.
\qed 

\begin{lemma}
  Let $\mathbb S = (\ul{\cat C},\monad T,\Ax)$ be a MES admitting free
  algebras.  Then, the quotient natural transformation 
$\qt^{\mathbb S}:T\natrightarrow T_{\mathbb S}$
is a strong functor morphism 
between the strong monads $\monad T$ and $\monad T_{\mathbb S}$.
That is, the following diagram commutes:
\[
\xymatrix@C=40pt{
V\tensor TX \ar[d]_-{\st_{V,X}} \ar[r]^-{V\tensor \qt^{\mathbb S}_X} & 
V\tensor T_{\mathbb S}X\ar[d]^-{\st^{\mathbb S}_{V,X}}\\ 
T(V\tensor X) \ar[r]_-{\qt^{\mathbb S}_{V\tensor X}} & T_{\mathbb S}(V\tensor X)\\
}
\]
\end{lemma}
\proof 
The commutativity of the diagram follows from the fact that both 
\[
\ol{\st^{\mathbb S}_{V,X} \comp (V\tensor \qt^{\mathbb S}_X)}\ ,\ 
\ol{\qt^{\mathbb S}_{V\tensor X} \comp \st_{V,X}}
\ :\ TX \rightarrow \inthom{V,T_\mathbb{S}(V\tensor X)}
\]
are an homomorphic extension of 
$\ol{\eta^{\mathbb S}_{V\tensor X}}: X \rightarrow
\inthom{V,T_\mathbb{S}(V\tensor X)}$
with respect to the Eilenberg-Moore algebra 
$(\inthom{V,T_\mathbb{S}(V\tensor X)}, {(\tau^\mathbb{S}_{V\tensor X})}^V)$
for the monad $\monad T$.
\qed 

We are now ready to prove the internal completeness theorem.
%
%
\begin{proof}[Proof of Theorem~\ref{thm:int-comp}]
We show \ref{thm:int-comp-1} $\Rightarrow$ \ref{thm:int-comp-2} $\Rightarrow$
\ref{thm:int-comp-3} $\Rightarrow$ \ref{thm:int-comp-1}.

\medskip
\noindent\textbf{\ref{thm:int-comp-1} $\Rightarrow$ \ref{thm:int-comp-2}.}
Holds vacuously.

\medskip
\noindent\textbf{\ref{thm:int-comp-2} $\Rightarrow$ \ref{thm:int-comp-3}.}
Because, for all $w: C\rightarrow TA$, the map 
$\qt^\mathbb{S}_A\comp w: C \rightarrow T_\mathbb{S} A$
factors as the composite
\[
\algstr{\llrrbrk{w}}{(T_\mathbb{S} A,\tau^\mathbb{S}_A)}
\comp (\widehat{\mathbf{p}}\tensor C)
\comp \lambda^{-1}_C
\]
for 
$
\xymatrix{
  \mathbf{p} =
  (I\tensor A 
  \ar[r]^-{\lambda_A} 
  & 
  A
  \ar[r]^-{\eta^{\mathbb S}_A}
  &  
  T_\mathbb{S}A ) }$.

\medskip
\noindent\textbf{\ref{thm:int-comp-3} $\Rightarrow$ \ref{thm:int-comp-1}.}
Because for every $(X,s:TX\rightarrow X)\in\catalg{\mathbb S}$ and
$w:C\rightarrow TA$, the interpretation map 
$\algstr{\llrrbrk{w}}{(X,s)}:[A,X]\tensor C\rightarrow X$ factors as the
composite 
\[
s^*
\comp
T_\mathbb{S}(\epsilon^A_X)
\comp
\st^\mathbb{S}_{[A,X],A}
\comp
\big([A,X]\tensor (\qt^\mathbb{S}_A \comp w)\big)
\]
where $s^*:T_\mathbb{S}X \rightarrow X$ is the unique homomorphic
extension of the identity map on $X$ with respect to the 
\mh{\mathbb S}algebra~$(X,s)$.  \end{proof}

\end{document}